\documentclass[journal]{IEEEtran}
\usepackage{amssymb}
\usepackage{amsmath,amssymb,amsfonts}
\usepackage{multirow}
\usepackage{graphicx}
\usepackage{textcomp}
\usepackage{amsthm}
\renewcommand{\vec}[1]{\boldsymbol{#1}}
\usepackage{setspace}
\usepackage{subfigure}

\usepackage{graphicx}
\usepackage{algorithm}
\usepackage{algorithmic}
\usepackage{amsmath}
\usepackage{color}
\usepackage{amsfonts}

\ifCLASSINFOpdf
\else
\fi

\newtheorem{thm}{Theorem}
\newtheorem{lem}{Lemma}

\newtheorem{defn}{Definition}

\newtheorem{rem}{Remark}
\newtheorem{pro}{Problem}

\begin{document}
%
% paper title
% Titles are generally capitalized except for words such as a, an, and, as,
% at, but, by, for, in, nor, of, on, or, the, to and up, which are usually
% not capitalized unless they are the first or last word of the title.
% Linebreaks \\ can be used within to get better formatting as desired.
% Do not put math or special symbols in the title.
\title{Continuous Profit Maximization: A Study of Unconstrained Dr-submodular Maximization}
%
%
% author names and IEEE memberships
% note positions of commas and nonbreaking spaces ( ~ ) LaTeX will not break
% a structure at a ~ so this keeps an author's name from being broken across
% two lines.
% use \thanks{} to gain access to the first footnote area
% a separate \thanks must be used for each paragraph as LaTeX2e's \thanks
% was not built to handle multiple paragraphs
%

\author{Jianxiong Guo,
	Weili Wu,~\IEEEmembership{Member,~IEEE}
	\thanks{J. Guo and W. Wu are with the Department
		of Computer Science, Erik Jonsson School of Engineering and Computer Science, Univerity of Texas at Dallas, Richardson, TX, 75080 USA
		
		E-mail: jianxiong.guo@utdallas.edu}% <-this 
	\thanks{Manuscript received April 19, 2005; revised August 26, 2015.}}

% note the % following the last \IEEEmembership and also \thanks - 
% these prevent an unwanted space from occurring between the last author name
% and the end of the author line. i.e., if you had this:
% 
% \author{....lastname \thanks{...} \thanks{...} }
%                     ^------------^------------^----Do not want these spaces!
%
% a space would be appended to the last name and could cause every name on that
% line to be shifted left slightly. This is one of those "LaTeX things". For
% instance, "\textbf{A} \textbf{B}" will typeset as "A B" not "AB". To get
% "AB" then you have to do: "\textbf{A}\textbf{B}"
% \thanks is no different in this regard, so shield the last } of each \thanks
% that ends a line with a % and do not let a space in before the next \thanks.
% Spaces after \IEEEmembership other than the last one are OK (and needed) as
% you are supposed to have spaces between the names. For what it is worth,
% this is a minor point as most people would not even notice if the said evil
% space somehow managed to creep in.

% The paper headers
\markboth{Journal of \LaTeX\ Class Files,~Vol.~14, No.~8, August~2015}%
{Shell \MakeLowercase{\textit{et al.}}: Bare Demo of IEEEtran.cls for IEEE Journals}
% The only time the second header will appear is for the odd numbered pages
% after the title page when using the twoside option.
% 
% *** Note that you probably will NOT want to include the author's ***
% *** name in the headers of peer review papers.                   ***
% You can use \ifCLASSOPTIONpeerreview for conditional compilation here if
% you desire.

% If you want to put a publisher's ID mark on the page you can do it like
% this:
%\IEEEpubid{0000--0000/00\$00.00~\copyright~2015 IEEE}
% Remember, if you use this you must call \IEEEpubidadjcol in the second
% column for its text to clear the IEEEpubid mark.

% use for special paper notices
%\IEEEspecialpapernotice{(Invited Paper)}

% make the title area
\maketitle

% As a general rule, do not put math, special symbols or citations
% in the abstract or keywords.
\begin{abstract}
	Profit maximization (PM) is to select a subset of users as seeds for viral marketing in online social networks, which balances between the cost and the profit from influence spread. We extend PM to that under the general marketing strategy, and form continuous profit maximization (CPM-MS) problem, whose domain is on integer lattices. The objective function of our CPM-MS is dr-submodular, but non-monotone. It is a typical case of unconstrained dr-submodular maximization (UDSM) problem, and take it as a starting point, we study UDSM systematically in this paper, which is very different from those existing researcher. First, we introduce the lattice-based double greedy algorithm, which can obtain a constant approximation guarantee. However, there is a strict and unrealistic condition that requiring the objective value is non-negative on the whole domain, or else no theoretical bounds. Thus, we propose a technique, called lattice-based iterative pruning. It can shrink the search space effectively, thereby greatly increasing the possibility of satisfying the non-negative objective function on this smaller domain without losing approximation ratio. Then, to overcome the difficulty to estimate the objective value of CPM-MS, we adopt reverse sampling strategies, and combine it with lattice-based double greedy, including pruning, without losing its performance but reducing its running time. The entire process can be considered as a general framework to solve the UDSM problem, especially for applying to social networks. Finally, we conduct experiments on several real datasets to evaluate the effectiveness and efficiency of our proposed algorithms.
\end{abstract}

% Note that keywords are not normally used for peerreview papers.

% For peer review papers, you can put extra information on the cover
% page as needed:
% \ifCLASSOPTIONpeerreview
% \begin{center} \bfseries EDICS Category: 3-BBND \end{center}
% \fi
%
% For peerreview papers, this IEEEtran command inserts a page break and
% creates the second title. It will be ignored for other modes.
\IEEEpeerreviewmaketitle

\section{Introduction}
\IEEEPARstart{O}{nline} social networks (OSNs) were becoming more and more popular to exchange ideas and make friends gradually in recent years, and accompanied by the rise of a series of social giants, such as Twitter, Facebook, Wechat, and LinkedIn. People tended to share what one sees and hears, and discuss some hot issues on these social platforms instead of traditional ways. Many companies or advertisers exploited to spread their products, opinions or innovations. By offering those influential users free or discounted samples, information can be spread across the whole network through word-of-mouth effect \cite{domingos2001mining} \cite{richardson2002mining}. Inspired from that, the influence maximization (IM) problem \cite{kempe2003maximizing} was formulated, which selects a subset of users (seed set) to maximizing the expected follow-up adoptions (influence spread) for a given information cascade. In this Kempe \textit{et al.}'s seminal work \cite{kempe2003maximizing}, IM was defined on the two basic discrete diffusion models, independent cascade model (IC-model) and linear threshold model (LT-model), and these two models can be generalized to the triggering model. Then, they proved the IM problem is NP-hard, and obtain a $(1-1/e)$-approximation under the IC/LT-model by use of a simple hill-climbing in the framework of monotonicity and submodularity.

Since this seminal work, plenty of related problems based on IM that used for different scenarios emerged \cite{guo2019novel} \cite{guo2019multi}. Among them, profit maximization (PM) \cite{lu2012profit} \cite{zhang2016profit} \cite{tang2016profit} \cite{tong2018coupon} \cite{8952599} is the most representative and widely used one. Consider viral marketing for a given product, the gain is the influence spread generated from our selected seed set in a social network. However, it is not free to activate those users in this seed set. For instance, in a real advertisement scenario, discounts and rewards are usually adopted to improve users' desire to purchase and stimulate consumption. Thus, the net profit is equal to influence spread minus the expense of seed set, where more incentives do not imply more benefit. Tang \textit{et al.} \cite{tang2016profit} proved the objective function of PM is submodular, but not monotone, and double greedy algorithm has a $(1/2)$-approximation if the objective value is non-negative. Before this, Kempe \textit{et al.} \cite{kempe2015maximizing} proposed the generalized marketing instead of the seed set. A marketing strategy is denoted by $\vec{x}\in\mathbb{Z}^d_+$ where a user $u$ will be activated as a seed with probability $h_u(\vec{x})$. Thus, the seed set is not deterministic, but activated probabilistically according to a marketing strategy. In this paper, we propose a continuous profit maximization under the general marketing strategies (CPM-MS) problem, which aims to choose the optimal marketing vector $\vec{x}^*\preceq\vec{b}$ such that the net profit can be maximized. Each component $\vec{x}(i)\in\vec{x}$ stand for the investment to marketing action $M_i$. Actually, in order to promote their products, a company often adopts multiple marketing techniques, such as advertisements, discounts, cashback, and propagandas, whose effects are different to customers at different levels. Therefore, CPM-MS is much more generalized than traditional PM.

In this paper, after formulating our CPM-MS problem, we discuss its properties first. We show that its objective function is NP-hard, and given a marketing vector $\vec{x}$, it is \#P-hard to compute the expected profit exactly. Because of the difficulty to compute the expected profit, we give an equivalent method that needs to run Monte-Carlo simulations on a constructed graph. Then, we prove that the objective function of CPM-MS problem is dr-submodular, but not monotone. Extended from set function to vector function on integer lattice, the dr-submodularity has a diminishing return property. For the unconstrained submodular maximization (USM), Buchbinder \textit{et al.} \cite{buchbinder2015tight} proposed a randomized double greedy algorithm that can achieve a tight $(1/2)$-approximation ratio. To our CPM-MS problem, we are able to consider it as a case of unconstrained dr-submodular maximization (UDSM) inspired by USM. Here, we introduce a lattice-based double greedy algorithm for the UDSM, and a $(1/2)$-approximation can be obtained as well if objective value is non-negative. The marketing vector $\vec{x}$ is defined on $\vec{0}\preceq\vec{x}\preceq\vec{b}$, thus this approximation can be guaranteed only when the sum of objective values on $\vec{0}$ and $\vec{b}$ is not less than zero, which is hard to be satisfied in the real applications. Imagine to offer all marketing actions full investments, is it still profitable? The answer is no. To overcome this defect, we design a lattice-based iterative pruning technique. It shrinks the searching space gradually in an iterative manner, and then we initialize our lattice-based double greedy with this smaller searching space. According to this revised process, the objective values on this smaller space are very likely to be non-negative, thereby increasing greatly the applicability of our algorithm's approximation. As mentioned earlier, even if we can use Monte-Carlo simulations to estimate the expected profit, its time complexity is too high, and thus restrict its scalability. Here, based on the reverse influence sampling (RIS) \cite{borgs2014maximizing} \cite{tang2014influence} \cite{tang2015influence} \cite{tang2018online}, we design an unbiased estimator for the profit function, which can estimate the objective value of a given marketing vector accurately if the number of samplings is large enough. Next, we take this estimator as our new objective function, combine with lattice-based pruning and double greedy algorithm, and propose DG-IP-RIS algorithm eventually. It guarantees to obtain a $(1/2-\varepsilon)$-approximation under a weak condition, whose time complexity is improved significantly. Finaly, we conduct several experiments to evaluate the superiority of our proposed DG-IP-RIS algorithm to other heuristic algorithms and compare their running times respectively, especially with the help of reverse sampling, which support the effectiveness and efficiency of our approaches strongly.

\textbf{Organization:} Sec. \uppercase\expandafter{\romannumeral2} discusses the related work. Sec. \uppercase\expandafter{\romannumeral3} is dedicated to introduce diffusion model, and formulate the problem based on that. The properties and computability of our CPM-MS problem are presented in Sec. \uppercase\expandafter{\romannumeral4}. Sec. \uppercase\expandafter{\romannumeral5} is the main contributions, including lattice-based double greedy and pruning algorithms. Sec. \uppercase\expandafter{\romannumeral6} analyzes the time complexity and designs speedup algorithms based on sampling strategies. Experiments and discussions are presented in Sec. \uppercase\expandafter{\romannumeral7} and \uppercase\expandafter{\romannumeral8} is the conclusion for this paper. 

\section{Related Work}
\textbf{Influence Maximization:} Kempe \textit{et al.} \cite{kempe2003maximizing} formulated IM to a combinatorial optimization problem, generalized the triggering model, including IC-model and LT-model, and proposed a greedy algorithm with $(1-1/e-\varepsilon)$-approximation by adopting Monte-Carlo simulations. Given a seed set, Chen \textit{et al.} proved that computing its exact influence spread under the IC-model \cite{chen2010scalable} and LT-model \cite{chen2010scal} are \#P-hard, and they designed two heuristic algorithms that can solve IM problem under the IC-model \cite{chen2010scalable} and LT-model \cite{chen2010scal}, which reduce the computation overhead effectively. Brogs \textit{et al.} \cite{borgs2014maximizing} took RIS to estimate the influence spread first, subsequently, a lot of researchers utilized RIS to design efficient algorithms with $(1-1/e-\varepsilon)$-approximation. Tang \textit{et al.} \cite{tang2014influence} proposed TIM/TIM+ algorithms, which were better than Brogs \textit{et al.}'s IM method regardless of accuracy and time complexity. Then, they developed a more efficient algorithm, IMM \cite{tang2015influence}, based on martingale analysis. Nguyen \textit{et al. } \cite{nguyen2016stop} designed SSA/D-SSA and claimed it reduces the running time significantly without losing approximation ratio, but still be doubted by other researchers. Recently, Tang \textit{et al.} \cite{tang2018online} created an online process of IM, and it can be terminated at any time and get a solution with its approximation guarantee.

\textbf{Profit Maximization:} Domingos \textit{et al.} \cite{domingos2001mining} \cite{richardson2002mining} studied viral marketing systematically first, and they proposed customers' value and utilized markov random fields to model the process of viral marketing. Lu \textit{et al.} \cite{lu2012profit} distinguished between influence and actual adoption, and designed a decision-making process to explain how to adopt a product. Zhang \textit{et al.} \cite{zhang2016profit} studied the problem of distributing a limited budget across multiple products such that maximizing total profit. Tang \textit{et al.} \cite{tang2016profit} analyzed and solved USM problem by double greedy algorithm thoroughly with PM as background, and proposed iterative pruning technique, which is different from our pruning process, because our objective function is defined on integer lattice. Tong \textit{et al.} \cite{tong2018coupon} considered the coupon allocation in the PM problem, and designed efficient randomized algorithms to achieve $(1/2-\varepsilon)$-approximation with high probability. Guo \textit{et al.} \cite{8952599} proposed a budgeted coupon problem whose domain is constrained, and provided a continuous double greedy algorithm with a valid approximation.

\textbf{(Dr-)submodular maximization:} Nemhauser \textit{et al.} \cite{nemhauser1978analysis} \cite{fisher1978analysis} began to study monotone submodular maximization problem, and laid the theoretical foundation, where IM problem was the most relevant work. However, PM is submodular, but not monotone, which is a case of USM \cite{feige2011maximizing} \cite{buchbinder2015tight}. Feige \textit{et al.} \cite{feige2011maximizing} pointed out no approximation algorithm exists for general USM unless giving an additional assumption, and they developed a deterministic local search wich $(1/3)$-approximation and a randomized local search with $(2/5)$-approximation for maximizing non-negative submodular function. Assume non-negativity satisfied as well, Buchbinder \textit{et al.} \cite{buchbinder2015tight} optimized it to $(1/2)$-approximation further with much lower computational complexity. Soma \textit{et al.} \cite{soma2015generalization} generalized the diminishing return property to the integer lattice first, and solved submodular cover problem with a bicriteria approximation algorithm. Then, they \cite{soma2018maximizing} studied monotone dr-submodular maximization problem exhaustively, where they designed algorithms with $(1-1/e)$-approximation under the cardinality, polymatroid, and knapsack constraint. Even if the UDSM problem was discussed in \cite{soma2017non}, their techniques and backgrounds were very different from ours. Therefore, how to solve UDSM, especially when non-negativity cannot be satisfied, is still an open problem, which is the main contribution of this paper.

\section{Problem Formulation}
In this section, we provides some preliminaries to the rest of this paper, and fomulate our continuous profit maximization problem under the general marketing strategies.
\subsection{Influence model}
An OSN can be abstracted as a directed graph $G=(V,E)$ where $V=\{v_1,v_2,\cdots,v_n\}$ is the set of $n$ nodes (users) and $E=\{e_1,e_2,\cdots,e_m\}$ is the set of $m$ edges (relationship between users). We default $|V|=n$ and $|E|=m$ when given $G=(V,E)$. For each directed edge $(u,v)\in E$, we say $v$ is an outgoing neighbor of $u$, and $u$ is an incoming neighbor of $v$. For any node $u\in V$, let $N^-(u)$ denote its set of incoming neighbors, and $N^+(u)$ denote its set of outgoing neighbors. In the process of influence diffusion, we consider a user is active if she accepts (is activated by) the information cascade from her neighbors or she is selected as a seed successfully. To model the influence diffusion, Kempe et al. \cite{kempe2003maximizing} proposed two classical models, IC-model and LT-model.

Let $S\subseteq V$ be a seed set and $S_i\subseteq V$ be the set of all active nodes at time step $t_i$. The influence diffusion initiated by $S$ can be represented by a discrete-time stochastic process. At time step $t_0$, all nodes in $S$ are activated, so we have $S_0:=S$. Under the IC-model, there is a diffusion probabiltiy $p_{uv}\in(0,1]$ associated with each edge $(u,v)\in E$. We set $S_i:=S_{i-1}$ at time step $t_i$ $(t\geq 1)$ first; then, for each node $u\in S_{i-1}\backslash S_{i-2}$, activated first at time step $t_{i-1}$, it have one chance to activate each of its inactive outgoing neighbor $v$ with probability $p_{uv}$. We add $v$ into $S_i$ if $u$ activates $v$ successfully at $t_i$. Under the LT-model, each edge $(u,v)\in E$ has a weight $b_{uv}$, and each node $v\in V$ has a threshold $\theta_v$ sampled uniformly in $[0,1]$ and $\sum_{u\in N^-(v)}b_{uv}\leq 1$. We set $S_i:=S_{i-1}$ at time step $t_i$ $(t\geq 1)$ first; then, for each inactive node $v\in V\backslash S_{i-1}$, it can be activated if $\sum_{u\in S_{i-1}\cap N^-(v)}b_{uv}\geq\theta_v$. We add $v$ into $S_i$ if $v$ is activated successfully at $t_i$. The influence diffusion terminates when no more inactive nodes can be activated. In this paper, we consider the triggering mode, where IC-model and LT-model are its special cases.

\begin{defn}[Triggering Model \cite{kempe2003maximizing}]
Each node $v$ selects a triggering set $T_v$ randomly and independently according to a distribution $\mathcal{D}_v$ over the subsets of $N^-(v)$. We set $S_i:=S_{i-1}$ at time step $t_i$ $(t\geq 1)$ first; then, for each inactive node $v\in V\backslash S_{i-1}$, it can be activated if there is at least one node in $T_v$ activated in $t_{i-1}$. We add $v$ into $S_i$ if $v$ is activated successfully at $t_i$. The influence diffusion terminates when no more inactive nodes can be activated.
\end{defn}

From above, a triggering model can be defined as $\Omega=(G,\mathcal{D})$, where $\mathcal{D}=\{\mathcal{D}_{v_1},\mathcal{D}_{v_2},\cdots\mathcal{D}_{v_n}\}$ is a set of distribution over the subsets of each $N^-(v_i)$.

\subsection{Realization}
For each node $v\in V$, under the IC-model, each node $u\in N^-(v)$ appears in $v$'s random triggering set $T_v$ with probability $p_{uv}$ independently. Under the LT-model, at most one node can appear in $T_v$, thus, for each node $u\in N^-(v)$, $T_v=\{u\}$ with probability $b_{uv}$ exclusively and $T_v=\emptyset$ with probability $1-\sum_{u\in N^-(v)}b_{uv}$. Now, we can define the realization (possible world) $g$ of graph $G$ under the triggering model $\Omega=(G,\mathcal{D})$, that is

\begin{defn}[Realization]
Given a directed graph $G=(V,E)$ and triggering model $\Omega=(G,\mathcal{D})$, a realization $g=\{T_{v_1},T_{v_2},\cdots,T_{v_n}\}$ of $G$ is a set of triggering set sampled from distribution $\mathcal{D}$, denoted by $g\sim\Omega$. For each node $v\in V$, we have $T_{v}\sim\mathcal{D}_{v}$ respectively.
\end{defn}

If a node $u$ appears in $v$'s triggering set, $u\in T_v$, we say edge $(u,v)$ is live, or else edge $(u,v)$ is blocked. Thus, realization $g$ can be regarded as a subgraph of $G$, which is the remaining graph by removing these blocked edges. Let $\Pr[g|g\sim\Omega]$ be the probability of realization $g$ of $G$ sampled from distribution $\mathcal{D}$, that is,
\begin{equation}
\Pr[g|g\sim\Omega]=\prod_{i=1}^{n}\Pr[T_{v_i}|T_{v_i}\sim\mathcal{D}_{v_i}]
\end{equation}
where $\Pr[T_{v_i}|T_{v_i}\sim\mathcal{D}_{v_i}]$ is the probability of $T_{v_i}$ sampled from $\mathcal{D}_{v_i}$. Under the IC-model, $\Pr[T_{v}|T_{v}\sim\mathcal{D}_{v}]=\prod_{u\in T_{v}}p_{uv}\prod_{u\in N^-(v)\backslash T_v}(1-p_{uv})$, and under the LT-model, $\Pr[T_{v}=\{u\}|T_{v}\sim\mathcal{D}_{v}]=b_{uv}$ for each $u\in N^-(v)$ and $\Pr[T_{v}=\emptyset|T_{v}\sim\mathcal{D}_{v}]=1-\sum_{u\in N^-(v)}b_{uv}$ deterministically.

Given a seed set $S\subseteq V$, we consider $I_\Omega(S)$ as a random variable that denotes the number of active nodes (influence spread) when the influence diffusion of $S$ terminates under the triggering model $\Omega=(G,\mathcal{D})$. Then, the number of nodes that are reachable from at least one node in $S$ under a realization $g$, $g\sim\Omega$, is denoted by $I_g(S)$. Thus, the expected influence spread $\sigma_\Omega(S)$, that is
\begin{equation}
\sigma_\Omega(S)=\mathbb{E}_{g\sim\Omega}[I_g(S)]=\sum_{g\sim\Omega}\Pr[g]\cdot I_g(S)
\end{equation}
where it is the weighted average of influence spread under all possible graph realizations. The IM problem aims to find a seed set $S$, such that $|S|\leq k$, to maximize the expected influence spread $\sigma_\Omega(S)$.
\begin{thm}[\cite{kempe2003maximizing}]
Under a triggering model $\Omega=(G,\mathcal{D})$, the expected influence spread $\sigma_\Omega(S)$ is monotone and submodular with respect to seed set $S$.
\end{thm}

\subsection{Problem Definition}
Under the general marketing strategies, the definition of IM problem will be different from above \cite{kempe2003maximizing}. Let $\mathbb{Z}^d_+$ be the collection of non-negative integer vector. A marketing strategy can be denoted by a $d$-dimensional vector $\vec{x}=(x_1,x_2,\cdots,x_d)\in\mathbb{Z}^d_+$, and we call it ``marketing vector''. Each component $\vec{x}(i)\in\mathbb{Z}_+$, $i\in[d]=\{1,2,\cdots,d\}$ means the number of investment units assigned to marketing action $M_i$. For example, $\vec{x}(i)=b$ tells us that marketing strategy $\vec{x}$ assigns $b$ investment units to marketing action $M_i$. Given a marketing vector $\vec{x}$, the probability that node $u\in V$ is activated as a seed is denoted by strategy function $h_u(\vec{x})$, where $h_u(\vec{x})\in[0,1]$. Thus, unlike the standard IM problem, the selection of seed set is not deterministic, but stochastic. Given a marketing vector $\vec{x}$, the probability of seed set $S$ sampled from $\vec{x}$, that is
\begin{equation}
\Pr[S|S\sim\vec{x}]=\prod_{u\in S}h_u(\vec{x})\cdot\prod_{v\in V\backslash S}(1-h_v(\vec{x}))
\end{equation}
where $\Pr[S|S\sim\vec{x}]$ is the probability that exactly nodes in $S$ are selected as seeds, but not in S are not selected as seeds under the marketing strategy $\vec{x}$, because each node is select as a seed independently. Thus, the expected influence spread $\mu_\Omega(\vec{x})$ of marketing vector $\vec{x}$ under the triggering model $\Omega(G,\mathcal{D})$ can be formulated, that is
\begin{flalign}
\mu_\Omega(\vec{x})&=\sum_{S\subseteq V}\Pr[S|S\sim\vec{x}]\cdot\sigma_\Omega(S)\\
&=\sum_{S\subseteq V}\sigma_\Omega(S)\cdot\prod_{u\in S}h_u(\vec{x})\cdot\prod_{v\in V\backslash S}(1-h_v(\vec{x}))
\end{flalign}

As we know, benefit is the gain obtained from influence spread and cost is the price required to pay for marketing strategy. Here, we assume each unit of marketing action $M_i$, $i\in[d]$, is associated with a cost $c_i\in\mathbb{R}_+$. Then, the total cost function $c:\mathbb{Z}^d_+\rightarrow\mathbb{R}_+$ can be defined as $c(\vec{x})=\sum_{i\in[d]}c_i\cdot\vec{x}(i)$. For simplicity, we consider the expected influence spread as our benefit. Thus, the expected profit $f_\Omega(\vec{x})$ we can obtain from marketing strategy $\vec{x}$ is the expected influence spread of $\vec{x}$ minus the cost of $\vec{x}$, that is
\begin{equation}
f_\Omega(\vec{x})=\mu_\Omega(\vec{x})-c(\vec{x})
\end{equation}
where $c(\vec{x})=\sum_{i\in[d]}c_i\cdot\vec{x}(i)$. Therefore, the continuous profit maximization under the general marketing strategies (CPM-MS) problem is formulated as follows:
\begin{pro}[CPM-MS]
Given a triggering model $\Omega=(G,\mathcal{D})$, a constraint vector $\vec{b}\in\mathbb{Z}^d_+$ and a cost function $c:\mathbb{Z}^d_+\rightarrow\mathbb{R}_+$, the CPM-MS problem aims to find an optimal marketing vector $\vec{x}^*\preceq\vec{b}$ that maximizes its expected profit $f_\Omega(\vec{x})$. That is, $\vec{x}^*=\arg\max_{\vec{x}\preceq\vec{b}}f_\Omega(\vec{x})$.
\end{pro}

\section{Properties of CPM-MS}
In this section, we introduce the submodularity on integer lattice, and then analyze the submodularity and computability of our CPM-MS problem.
\subsection{Submodularity on Integer Lattice}
Generally, defined on set, a set function $\alpha:2^V\rightarrow\mathbb{R}$ is monotone if $\alpha(S)\leq \alpha(T)$ for any $S\subseteq T\subseteq V$, and submodular if $\alpha(S)+\alpha(T)\geq \alpha(S\cup T)+\alpha(S\cap T)$. The submodularity of set function implies a diminishing return property, thus $\alpha(S\cup\{u\})-\alpha(S)\geq \alpha(T\cup\{u\})-\alpha(T)$ for any $S\subseteq T\subseteq V$ and $u\notin T$. These two definitions of submodularity on set function are equivalent. Defined on integer lattice, a vector function $\beta:\mathbb{Z}^d_+\rightarrow\mathbb{R}$ is monotone if $\beta(\vec{s})\leq \beta(\vec{t})$ for any $\vec{s}\preceq\vec{t}\in\mathbb{Z}^d_+$, and submodular if $\beta(\vec{s})+\beta(\vec{t})\geq \beta(\vec{s}\lor\vec{t})+\beta(\vec{s}\land\vec{t})$ for any $s,t\in\mathbb{Z}^d_+$, where $(\vec{s}\lor\vec{t})(i)=\max\{\vec{s}(i),\vec{t}(i)\}$ and $(\vec{s}\land\vec{t})(i)=\min\{\vec{s}(i),\vec{t}(i)\}$. Here, $\vec{s}\preceq\vec{t}$ implies $\vec{s}(i)\leq\vec{t}(i)$ for each component $i\in[d]$. Besides, we consider a vector function is diminishing return submodular (dr-submodular) if $\beta(\vec{s}+\vec{e}_i)-\beta(\vec{s})\geq \beta(\vec{t}+\vec{e}_i)-\beta(\vec{t})$ for any $\vec{s}\preceq\vec{t}$ and $i\in [d]$, where $\vec{e}_i\in\mathbb{Z}^d_+$ is the $i$-th unit vector with the $i$-th component being $1$ and others being $0$. Different from the submodularity for a set function, for a vector function, $\beta$ is submodular does not mean it is dr-submodular, but the oppposite is true. Thus, dr-submodularity is stronger than submodularity generally.
\begin{lem}
Given a set function $\alpha:2^V\rightarrow\mathbb{R}$ and a vector function $\beta:\mathbb{Z}^d_+\rightarrow\mathbb{R}$, they sastisfy
\begin{equation}
\beta(\vec{x})=\sum_{S\subseteq V}\alpha(S)\cdot\prod_{u\in S}h_u(\vec{x})\cdot\prod_{v\in V\backslash S}(1-h_v(\vec{x}))
\end{equation}
If $\alpha(\cdot)$ is monotone and submodular and $h_u(\cdot)$ is monotone and dr-submodular for each $u\in V$, then $\beta(\cdot)$ is monotone and dr-submodular.
\end{lem}
\begin{proof}
It is an indirect corollary that has been implied by the proof process of section 7 in \cite{kempe2015maximizing} and \cite{guo2020continuous}.
\end{proof}
\begin{thm}
Given a triggering model $\Omega=(G,\mathcal{D})$, the profit function $f_\Omega(\cdot)$ is dr-submodular, but not monotone.
\end{thm}
\begin{proof}
From Lemma 1, Theorem 1, and Equation (5), we have known that the expected influence spread $\mu_\Omega(\cdot)$ is monotone and dr-submodular because $\sigma_\Omega(\cdot)$ is monotone and submodular. Thus, we have $f_\Omega(\vec{x}+\vec{e}_i)-f_\Omega(\vec{x})=\mu_\Omega(\vec{x}+\vec{e}_i)-\mu_\Omega(\vec{x})-c_i\geq\mu_\Omega(\vec{y}+\vec{e}_i)-\mu_\Omega(\vec{y})-c_i=f_\Omega(\vec{y}+\vec{e}_i)-f_\Omega(\vec{y})$ iff $\vec{x}\preceq\vec{y}\in\mathbb{Z}^d_+$. Thus, $f_\Omega(\cdot)$ is dr-submodular.
\end{proof}

\subsection{Computability}
Given a seed set $S\subseteq V$, it is \#P-hard to compute the expected influence spread $\sigma_\Omega(S)$ under the IC-model \cite{chen2010scalable} and the LT-model \cite{chen2010scal}. Assume that a marketing vector $\vec{x}\in\{0,1\}^n$ and $h_u(\vec{x})=\vec{x}(u)$ for $u\in V$ where user $u$ is a seed if and only if $\vec{x}(u)=1$. According to the Equation (4), the expected influence spread $\mu_\Omega(\vec{x})$ is equivalent to $\sigma_\Omega(S)$ in which $S=\{u\in V:\vec{x}(u)=1\}$. Thereby, given a marketing vector $\vec{x}$, computing the expected influence spread $\mu_\Omega(\vec{x})$ is \#P-hard as well under the IC-model and LT-model. Subsequently, a natural question how to estimate the value of $\mu_\Omega(\vec{x})$ given $\vec{x}$ effectively. To estimate $\mu_\Omega(\vec{x})$, we usually adopt Monte-Carlo simulations. However, it is inconvenient for us to use such a method here because the randomness comes from two parts, one is from the seed selection, and the other is from the process of influence diffusion. Therefore, we require to design a more simple and efficient method.

First, we are able to establish an equivalent relationship between $\sigma_\Omega(\cdot)$ and $\mu_\Omega(\cdot)$. Given a social network $G=(V,E)$ and a marketing vector $\vec{x}\in\mathbb{Z}^d_+$, we create a constructed graph $\widetilde{G}=(\widetilde{V},\widetilde{E})$ by adding a new node $\widetilde{u}$ and a new directed edge $(\widetilde{u},u)$ for each node $u\in V$ to $G$. Take IC-model for instance, the diffusion probability for this new edge $(\widetilde{u},u)$ can be set as $p_{\widetilde{u}u}=h_u(\vec{x})$. Then, we have
\begin{equation}
	\mu(\vec{x}|G)=\sigma(\widetilde{V}-V|\widetilde{G})-|V|
\end{equation}
where $\mu(\cdot|G)$ and $\sigma(\cdot|\widetilde{G})$ imply that we compute them under the graph $G$ and the constructed graph $\widetilde{G}$.
\begin{thm}
	Given a social network $G=(V,E)$ and a marketing vector $\vec{x}\in\mathbb{Z}^d_+$, the expected influence spread $\mu_\Omega(\vec{x})$ can be estimated with $(\gamma,\delta)$-approximation by Monte-Carlo simulations in $O\left(\frac{(m+3n)n^2\ln(2/\delta)}{2(\gamma\sum_{u\in V}h_u(\vec{x}))^2}\right)$ running time.
\end{thm}
\begin{proof}
Mentioned above, we can compute $\sigma(\widetilde{V}-V|\widetilde{G})$ on the constructed graph instead of $\mu(\vec{x}|G)$ on the original graph. Let $S=\widetilde{V}-V$, the value of $\sigma_\Omega(S)$ can be estimated by Monte-Carlo simulations according to (2). Based on Hoeffding's inequality, we can note that
\begin{equation*}
	\Pr\left[\left|\hat{\sigma}_\Omega(S)-\sigma_\Omega(S)\right|\geq\gamma(\sigma_\Omega(S)-n)\right]\leq 2e^{-2r\left(\frac{\gamma(\mu_\Omega(S)-n)}{n}\right)^2}
\end{equation*}
where $r$ is the number of Monte-Carlo simulations and $\sigma_\Omega(S)-n\leq n$. We have $\sigma_\Omega(S)-n\geq\sum_{u\in V}h_u(\vec{x})$, and to achieve a $(\gamma,\delta)$-estimation, the number of Monte-Carlo simulations $r\geq\frac{n^2\ln(2/\delta)}{2(\gamma\sum_{u\in V}h_u(\vec{x}))^2}$. For each iteration of simulations in the constructed graph, it takes $O(m+3n)$ running time. Thus, we can obtain a $(\gamma,\delta)$-approximation of $\mu_\Omega(\vec{x})$ in $O\left(\frac{(m+3n)n^2\ln(2/\delta)}{2(\gamma\sum_{u\in V}h_u(\vec{x}))^2}\right)$ running time.
\end{proof}

Based on the Theorem 3, we can get an accurate estimation for the objective function $f_\Omega(\vec{x})$, shown as (6), of CPM-MS problem by adjusting the parameter $\gamma$ and $\delta$ definitely.

\section{Algorithms Design}
From the last section, we have known that the objective function of CPM-MS is dr-submodular, but not monotone. In this section, we develop our new methods based on the double greedy algorithm \cite{buchbinder2015tight} for our CPM-MS, and obtain an optimal approximation ratio.

\subsection{Lattice-based Double Greedy}
For non-negative submodular functions, Buchbinder \textit{et al.} \cite{buchbinder2015tight} designed a double greedy algorithm to get a solution for the USM problem with a tight theoretical guarantee. Under the deterministic setting, the double greedy algorithm has a $(1/3)$-approximation, while it has a $(1/2)$-approximation under the randomized setting. Extending from set to integer lattice, we derive a revised double greedy algorithm that is suitable for dr-submodular functions, namely UDSM problem. We adopt the randomized setting, and the lattice-based double greedy algorithm is shown in Algorithm \ref{a1}. We omit the subscript of $f_\Omega(\cdot)$, denote it by $f(\cdot)$ from now on.

\begin{algorithm}[!t]
	\caption{\text{Lattice-basedDoubleGreedy}}\label{a1}
	\begin{algorithmic}[1]
		\renewcommand{\algorithmicrequire}{\textbf{Input:}}
		\renewcommand{\algorithmicensure}{\textbf{Output:}}
		\REQUIRE $f:\mathbb{Z}^d_+\rightarrow\mathbb{R}$, $[\vec{s},\vec{t}]$ where $\vec{s}\preceq\vec{t}\in\mathbb{Z}_+^d$
		\ENSURE $\vec{x}\in\mathbb{Z}^d_+$
		\STATE Initialize: $\vec{x}\leftarrow\vec{s}$, $\vec{y}\leftarrow\vec{t}$
		\FOR {$i\in[d]$}
		\WHILE {$\vec{x}(i)<\vec{y}(i)$}
		\STATE $a\leftarrow f(\vec{e}_i|\vec{x})$ and $b\leftarrow f(-\vec{e}_i|\vec{y})$
		\STATE $a'\leftarrow\max\{a,0\}$ and $b'\leftarrow\max\{b,0\}$
		\STATE $r\leftarrow\text{Uniform}(0,1)$
		\STATE (Note: we set $a'/(a'+b')=1$ if $a'=b'=0$)
		\IF {$r\leq a'/(a'+b')$}
		\STATE $\vec{x}\leftarrow\vec{x}+\vec{e}_i$ and $\vec{y}\leftarrow\vec{y}$
		\ELSE
		\STATE $\vec{y}\leftarrow\vec{y}-\vec{e}_i$ and $\vec{x}\leftarrow\vec{x}$
		\ENDIF
		\ENDWHILE
		\ENDFOR
		\RETURN $\vec{x}(=\vec{y})$\
	\end{algorithmic}
\end{algorithm}

Here, we denote by $f(\vec{e}_i|\vec{x})=f(\vec{x}+\vec{e}_i)-f(\vec{x})$, the marginal gain of adding component $i\in[d]$ by $1$. Generally, this algorithm is initialized by $[\vec{0},\vec{b}]$, and for each component $i\in[d]$, we increase $\vec{x}(i)$ by $1$ or decrease $\vec{y}(i)$ by $1$ until they are equal in each inner (while) iteration. The result returned by Algorithm \ref{a1} has $\vec{x}=\vec{y}$. Then, we have the following conclusion which can be inferred directly from double greedy algorithm in \cite{buchbinder2015tight}, that is
\begin{thm}
For our CPM-MS problem, if we initalize Algorithm \ref{a1} by $[\vec{0},\vec{b}]$, and $f(\vec{0})+f(\vec{b})\geq0$ is satisfied, the marketing vector $\vec{x}^\circ$ returned by Algorithm \ref{a1} is a $(1/2)$-approximate solution, such that
\begin{equation}
	\mathbb{E}[f(\vec{x}^\circ)]\geq(1/2)\cdot\max_{\vec{x}\preceq\vec{b}}f(\vec{x})
\end{equation}
\end{thm}
\noindent
Here, $f(\vec{x})\geq0$ for any $\vec{x}\preceq\vec{b}$ is equivalent to say $f(\vec{0})+f(\vec{b})\geq0$, namely $f(\vec{b})\geq0$ because of $f(\vec{0})=0$, which is a natural inference from the dr-submodularity.

\begin{algorithm}[!t]
	\caption{\text{Lattice-basedPruning}}\label{a2}
	\begin{algorithmic}[1]
		\renewcommand{\algorithmicrequire}{\textbf{Input:}}
		\renewcommand{\algorithmicensure}{\textbf{Output:}}
		\REQUIRE $f:\mathbb{Z}^d_+\rightarrow\mathbb{R}$, $\vec{b}\in\mathbb{Z}^d_+$
		\ENSURE $\pi_t=[\vec{g}_t,\vec{h}_t]$
		\STATE Initialize: $\vec{g}_t\leftarrow\vec{0}$, $\vec{h}_t\leftarrow\vec{b}$
		\STATE Initialize: $t\leftarrow 0$
		\WHILE {$\vec{g}_t\neq\vec{g}_{t-1}$ or $\vec{h}_t\neq\vec{h}_{t-1}$}
		\FOR {$i\in[d]$}
		\IF {$\vec{g}_t(i)=\vec{h}_t(i)$}
		\STATE $\vec{g}_{t+1}(i)\leftarrow\vec{g}_t(i)$
		\STATE $\vec{h}_{t+1}(i)\leftarrow\vec{h}_t(i)$
		\STATE Continue
		\ENDIF
		\IF {$f(\vec{e}_i|\vec{h}_t-\vec{h}_t(i)\vec{e}_i+\vec{g}_t(i)\vec{e}_i)\leq0$}
		\STATE $\vec{g}_{t+1}(i)\leftarrow\vec{g}_t(i)$
		\ELSE
		\STATE $\vec{g}_{t+1}(i)\leftarrow\vec{g}_t(i)+\max\{k:f(\vec{e}_i|\vec{h}_t-\vec{h}_t(i)\vec{e}_i+\vec{g}_t(i)\vec{e}_i+(k-1)\vec{e}_i)>0\}$, $k\in\{1,\cdots,\vec{h}_t(i)-\vec{g}_t(i)\}$
		\ENDIF
		\IF {$f(\vec{e}_i|\vec{g}_t)<0$}
		\STATE $\vec{h}_{t+1}(i)\leftarrow\vec{h}_t(i)$
		\ELSE
		\STATE $\vec{h}_{t+1}(i)\leftarrow\vec{g}_t(i)+\max\{k:f(\vec{e}_i|\vec{g}_t+(k-1)\vec{e}_i)\geq 0\}$, $k\in\{0,\cdots,\vec{h}_t(i)-\vec{g}_t(i)\}$
		\ENDIF
		\ENDFOR
		\STATE $t\leftarrow t+1$
		\ENDWHILE
		\RETURN $\pi_t=[\vec{g}_t,\vec{h}_t]$
	\end{algorithmic}
\end{algorithm}

\subsection{Lattice-based Iterative Pruning}
According to Theorem 4, $(1/2)$-approximation is based on an assumption that $f(\vec{0})+f(\vec{b})\geq0$, and this is almost impossible in may real application scenarios. It means that we are able to gain profit if giving all marketing action full investments, which is ridiculous for viral marketing. However, a valid approximation ratio cannot be obtained by using Algorithm \ref{a1} when $f(\vec{b})<0$ exists. To address this problem, Tang \textit{et al.} \cite{tang2016profit} proposed a groundbreaking techniques, called iterative pruning, to reduce the search space such that the objective is non-negative in this space and without losing approximation guarantee. But their techniques are designed for submodular function based on set domain, it cannot be applied to dr-submodular function directly. Inspired by \cite{tang2016profit}, we develop an iterative pruning technique suitable for dr-submodular functions in this section, which is a non-trival transformation from set to integer lattice.

Given a dr-submodular function $f(\vec{x})$ defined on $\vec{x}\preceq\vec{b}$, we have two vectors $\vec{g}_1$ and $\vec{h}_1$, such that: (1) $\vec{g}_1(i)=0$ if $f(\vec{e}_i|\vec{b}-\vec{b}(i)\vec{e}_i)\leq0$, or else $\vec{g}_1(i)=\max\{k:f(\vec{e}_i|\vec{b}-\vec{b}(i)\vec{e}_i+(k-1)\vec{e}_i)>0\}$ for $k\in\{1,\cdots,\vec{b}(i)\}$; (2) $\vec{h}_1(i)=0$ if $f(\vec{e}_i|\vec{0})<0$, or else $\vec{h}_1(i)=\max\{k:f(\vec{e}_i|\vec{0}+(k-1)\vec{e}_i)\geq0\}$ for $k\in\{1,\cdots,\vec{b}(i)\}$.

\begin{lem}
	We have $\vec{g}_1\preceq\vec{h}_1$.
\end{lem}
\begin{proof}
	For any component $i\in[d]$, we have $f(\vec{e}_i|\vec{b}-\vec{b}(i)\vec{e}_i+\vec{g}_1(i)\vec{e}_i)\leq0$ but $f(\vec{e}_i|\vec{b}-\vec{b}(i)\vec{e}_i+(\vec{g}_1(i)-1)\vec{e}_i)>0$; and $f(\vec{e}_i|\vec{0}+\vec{h}_1(i)\vec{e}_i)<0$ but $f(\vec{e}_i|\vec{0}+(\vec{h}_1(i)-1)\vec{e}_i)\geq0$. Because of dr-submodularity, it satisfies $f(\vec{e}_i|\vec{b}-\vec{b}(i)\vec{e}_i+\vec{h}_1(i)\vec{e}_i)\leq f(\vec{e}_i|\vec{0}+\vec{h}_1(i)\vec{e}_i)<0$. Thus, $(\vec{g}_1(i)-1)<\vec{h}_1(i)$ and $\vec{g}_1(i)\leq\vec{h}_1(i)$. Subsequently, $\vec{g}_1\preceq\vec{h}_1$.
\end{proof}
\noindent
Then, we define a collection denoted by $\pi_1=[\vec{g}_1,\vec{h}_1]$ that contains all the marketing vectors $\vec{x}$ that satisfies $\vec{g}_1\preceq\vec{h}_1$. Apparently, $\pi_1$ is a subcollection of $[\vec{0},\vec{b}]$.
\begin{lem}
All optimal solutions $\vec{x}^*$ that satisfy $f(\vec{x}^*)=\max_{\vec{x}\preceq\vec{b}}f(\vec{x})$ are contained in the collection $\pi_1=[\vec{g}_1,\vec{h}_1]$, i.e., $\vec{g}_1\preceq\vec{x}^*\preceq\vec{h}_1$ for all $\vec{x}^*$.
\end{lem}
\begin{proof}
	For any component $i\in[d]$, we consider any vector $\vec{x}$ with $\vec{x}(i)<\vec{g}_1(i)$, we have $f(\vec{e}_i|\vec{x})\geq f(\vec{e}_i|\vec{b}-\vec{b}(i)\vec{e}_i+\vec{x}(i)\vec{e}_i)>0$ because of dr-submodularity. Thereby, $\vec{x}+\vec{e}_i$ has a larger profit than $\vec{x}$ for sure, so the $i$-th component of the optimal marketing vector $\vec{x}^*$ at least eqauls to $\vec{x}(i)+1$, which indicates $\vec{x}^*(i)\geq\vec{g}_1(i)$. On the other have, consider $\vec{x}(i)
	\geq\vec{h}_1(i)$, we have $f(\vec{e}_i|\vec{x})\leq f(\vec{e}_i|\vec{0}+\vec{x}(i)\vec{e}_i)<0$. Thereby, $\vec{x}+\vec{e}_i$ has a less profit than $\vec{x}$ for sure, so the $i$-th component of the optimal marketing vector $\vec{x}^*$ at most equals to $\vec{x}(i)$, which indicates $\vec{x}^*(i)\leq\vec{h}_1(i)$. Thus, $\vec{g}_1\preceq\vec{x}^*\preceq\vec{h}_1$.
\end{proof}

From above, Lemma 3 determines a range for the optimal vector, thus reducing the search space. Then, the collection $\pi_1=[\vec{g}_1,\vec{h}_1]$ can be pruned further in an iterative manner. Now, the upper bound of the optimal vector is $\vec{h}_1$, i.e., $\vec{x}^*\preceq\vec{h}_1$, hereafter, we are able to increase $\vec{g}_1$ to $\vec{g}_2$, where $\vec{g}_2(i)=\vec{g}_1(i)$ if $f(\vec{e}_i|\vec{h}_1-\vec{h}_1(i)\vec{e}_i+\vec{g}_1(i)\vec{e}_i)\leq0$, or else $\vec{g}_2(i)=\vec{g}_1(i)+\max\{k:f(\vec{e}_i|\vec{h}_1-\vec{h}_1(i)\vec{e}_i+\vec{g}_1(i)\vec{e}_i+(k-1)\vec{e}_i)>0\}$ for $k\in\{1,\cdots,\vec{h}_1(i)-\vec{g}_1(i)\}$. The lower bound of the optimal vector is $\vec{g}_1$, i.e., $\vec{x}^*\succeq\vec{g}_1$, and similarly we are able to decrease $\vec{h}_1$ to $\vec{h}_2$, where $\vec{h}_2(i)=\vec{h}_1(i)$ if $f(\vec{e}_i|\vec{g}_1)<0$, or else $\vec{h}_2(i)=\vec{g}_1(i)+\max\{k:f(\vec{e}_i|\vec{g}_1+(k-1)\vec{e}_i)\geq 0\}$ for $k\in\{1,\cdots,\vec{h}_1(i)-\vec{g}_1(i)\}$. In this process, it generates a more compressed collection $\pi_2=[\vec{g}_2,\vec{h}_2]$ than $\pi_1$. We repeat this process iteratively until $\vec{g}_t$ and $\vec{h}_t$ cannot be increased and decreased further. The Lattice-basedPruning algorithm is shown in Algorithm \ref{a2}. The collection returned by Algorithm \ref{a2} is denoted by $\pi^\circ=[\vec{g}^\circ,\vec{h}^\circ]$.

\begin{lem}
All optimal solutions $\vec{x}^*$ that satisfy $f(\vec{x}^*)=\max_{\vec{x}\preceq\vec{b}}f(\vec{x})$ are contained in the collection $\pi^\circ=[\vec{g}^\circ,\vec{h}^\circ]$, and $\vec{g}_t\preceq\vec{g}_{t+1}\preceq\vec{g}^\circ\preceq\vec{x}^*\preceq\vec{h}^\circ\preceq\vec{h}_{t+1}\preceq\vec{h}_t$ holds for all $\vec{x}^*$ and any $t\geq 0$.
\end{lem}
\begin{proof}
	First, we show that the collection generated in current iteration is a subcollection of that generated in previous iteration, namely $\vec{g}_t\preceq\vec{g}_{t+1}\preceq\vec{h}_{t+1}\preceq\vec{h}_t$. We prove it by induction. In Lemma 2, we have shown that $\vec{g}_0=\vec{0}\preceq\vec{g}_1\preceq\vec{h}_1\preceq\vec{h}_0=\vec{b}$. For any $t>1$, we assume that $\vec{g}_{t-1}\preceq\vec{g}_t\preceq\vec{h}_t\preceq\vec{h}_{t-1}$ is satisfied. Given a component $i\in[d]$, for any $q\leq\vec{g}_t(i)$, we have $f(\vec{e}_i|\vec{h}_{t-1}-\vec{h}_{t-1}(i)\vec{e}_i+(q-1)\vec{e}_i)>0$. Because of the dr-submodularity, we have $f(\vec{e}_i|\vec{h}_{t}-\vec{h}_{t}(i)\vec{e}_i+(q-1)\vec{e}_i)\geq f(\vec{e}_i|\vec{h}_{t-1}-\vec{h}_{t-1}(i)\vec{e}_i+(q-1)\vec{e}_i)>0$, which indicates $\vec{g}_t\preceq\vec{g}_{t+1}$. Similarly, for any $q\leq\vec{h}_{t+1}(i)$, we have $f(\vec{e}_i|\vec{g}_t-\vec{g}_t(i)\vec{e}_i+(q-1)\vec{e}_i)\geq0$. Because of the dr-submodularity, we have $f(\vec{e}_i|\vec{g}_{t-1}-\vec{g}_{t-1}(i)\vec{e}_i+(q-1)\vec{e}_i)\geq f(\vec{e}_i|\vec{g}_{t}-\vec{g}_{t}(i)\vec{e}_i+(q-1)\vec{e}_i)\geq0$, which indicates $\vec{h}_{t+1}\preceq\vec{h}_t$. Moreover, for any $q\leq\vec{g}_{t+1}(i)$, we have $f(\vec{e}_i|\vec{h}_{t}-\vec{h}_{t}(i)\vec{e}_i+(q-1)\vec{e}_i)>0$. Due to $\vec{g}_t\preceq\vec{h}_t$ and dr-submodularity, we have $f(\vec{e}_i|\vec{g}_{t}-\vec{g}_{t}(i)\vec{e}_i+(q-1)\vec{e}_i)\geq f(\vec{e}_i|\vec{h}_{t}-\vec{h}_{t}(i)\vec{e}_i+(q-1)\vec{e}_i)>0$, which indicates $\vec{g}_{t+1}\preceq\vec{h}_{t+1}$. Thus, we conclude that  $\vec{g}_t\preceq\vec{g}_{t+1}\preceq\vec{h}_{t+1}\preceq\vec{h}_t$ holds for any $t\geq 0$.
	
	Then, we show that any optimal solutions $\vec{x}^*$ are contained in the collection $\pi^\circ=[\vec{g}^\circ,\vec{h}^\circ]$ returned by Algorithm \ref{a2}, namely $\vec{g}^\circ\preceq\vec{x}^*\preceq\vec{h}^\circ$. We prove it by induction. In Lemma 3, we have shown that $\vec{g}_1\preceq\vec{x}^*\preceq\vec{h}_1$. For any $t>1$, we assume that $\vec{g}_t\preceq\vec{x}^*\preceq\vec{h}_t$ is satisfied. Given a component $i\in[d]$, for any $q\leq\vec{g}_{t+1}(i)$, we have $f(\vec{e}_i|\vec{h}_t-\vec{h}_t(i)\vec{e}_i+(q-1)\vec{e_i})>0$. Because of the dr-submodularity, we have $f(\vec{e}_i|\vec{x}^*-\vec{x}^*(i)\vec{e}_i+(q-1)\vec{e_i})\geq f(\vec{e}_i|\vec{h}_t-\vec{h}_t(i)\vec{e}_i+(q-1)\vec{e_i})>0$, which implies $\vec{x}^*(i)\geq q$. Otherwise, if $\vec{x}^*(i)<q$, we have $f(\vec{e}_i|\vec{x}^*)>0$, which contradicts the optimality of $\vec{x}^*$, thus $\vec{x}^*\succeq\vec{g}_{t+1}$. Similarly, for any $q>\vec{h}_{t+1}(i)$, we have $f(\vec{e}_i|\vec{g}_t-\vec{g}_t(i)\vec{e}_i+(q-1)\vec{e}_i)<0$. Because of the dr-submodularity, we have 
	$f(\vec{e}_i|\vec{x}^*-\vec{x}^*(i)\vec{e}_i+(q-1)\vec{e_i})<f(\vec{e}_i|\vec{g}_t-\vec{g}_t(i)\vec{e}_i+(q-1)\vec{e}_i)<0$, which implies $\vec{x}^*(i)<q$. Otherwise, if $\vec{x}^*(i)\geq q$, we have $f(\vec{e}_i|\vec{x}^*-\vec{x}^*(i)\vec{e}_i+(q-1)\vec{e}_i)<0$, which contradicts the optimality of $\vec{x}^*$, thus $\vec{x}^*\preceq\vec{h}_{t+1}$. Thus, we conclude that $\vec{g}_{t+1}\preceq\vec{x}^*\preceq\vec{h}_{t+1}$ holds for any $t\geq 0$, and $\vec{g}^\circ\preceq\vec{x}^*\preceq\vec{h}^\circ$. The proof of Lemma is completed.
\end{proof}
\begin{lem}
For any two vectors $\vec{x},\vec{y}\in\mathbb{Z}^d_+$ with $\vec{x}\preceq\vec{y}$ and a dr-submodular function $f:\mathbb{Z}^d_+\rightarrow\mathbb{R}$, we have
\begin{equation}
f(\vec{y})=f(\vec{x})+\sum_{i=1}^{d}\sum_{j=1}^{\vec{z}(i)}f\Bigg(\vec{e}_i|\vec{x}+\sum_{k=1}^{i-1}\sum_{l=1}^{\vec{z}(k)}\vec{e}_k+\sum_{l=1}^{j-1}\vec{e}_i\Bigg)
\end{equation}
where we define $\vec{z}(i)=\vec{y}(i)-\vec{x}(i)$.
\end{lem}

To understand Lemma 5, we give a simple example here. Let vector $\vec{x},\vec{y}$ be $\vec{x}=(1,1),\vec{y}=(2,3)$, subsequently we can see $\vec{x}\preceq\vec{y}$ and $\vec{z}=(1,2)$. From the definition of (10), we have $f(\vec{x})+f(\vec{e}_1|\vec{x})+f(\vec{e}_2|\vec{x}+\vec{e}_1)+f(\vec{e}_2|\vec{x}+\vec{e}_1+\vec{e}_2)=f(\vec{x}+\vec{e}_1+2\vec{e}_2)=f(\vec{y})$, which reflects the essence and correctness of Lemma 5 definitely.

\begin{lem}
The $f(\vec{g}_t)$ and $f(\vec{h}_t)$ are monotone non-decreasing with the increase of $t$.
\end{lem}
\begin{proof}
We prove that $f(\vec{g}_t)\leq f(\vec{g}_{t+1})$ and $f(\vec{h}_t)\leq f(\vec{h}_{t+1})$ respectively. Given a component $i\in[d]$, for any $q\leq\vec{g}_{t+1}(i)$, we have $f(\vec{e}_i|\vec{h}_t-\vec{h}_t(i)\vec{e}_i+(q-1)\vec{e}_i)>0$. Because of the dr-submodularity, we have $f(\vec{e}_i|\vec{g}_{t+1}-\vec{g}_{t+1}(i)\vec{e}_i+(q-1)\vec{e}_i)\geq  f(\vec{e}_i|\vec{h}_t-\vec{h}_t(i)\vec{e}_i+(q-1)\vec{e}_i)>0$, where $\vec{g}_{t+1}\preceq\vec{h}_t$. According to the Lemma 5, that is 
\begin{flalign}
&f(\vec{g}_{t+1})=f(\vec{g}_t)+\sum_{i=1}^{d}\sum_{j=1}^{\vec{z}_t(i)}f\Bigg(\vec{e}_i|\vec{g}_t+\sum_{k=1}^{i-1}\sum_{l=1}^{\vec{z}_t(k)}\vec{e}_k+\sum_{l=1}^{j-1}\vec{e}_i\Bigg)\nonumber\\
&\geq f(\vec{g}_t)+\sum_{i=1}^{d}\sum_{j=1}^{\vec{z}_t(i)}f\Bigg(\vec{e}_i|\vec{g}_{t+1}-\vec{z}_t(i)\vec{e}_i+\sum_{l=1}^{j-1}\vec{e}_i\Bigg)
\end{flalign}
where $\vec{z}_t(i)=\vec{g}_{t+1}(i)-\vec{g}_t(i)$. The inequality (11) is established since its dr-submodularity, that is $\vec{g}_t+\sum_{k=1}^{i-1}\sum_{l=1}^{\vec{z}_t(k)}\vec{e}_k\preceq\vec{g}_{t+1}-\vec{z}_t(i)\vec{e}_i$ definitely. Besides, since $f(\vec{e}_i|\vec{g}_{t+1}-\vec{z}_t(i)\vec{e}_i+\sum_{l=1}^{j-1}\vec{e}_i)>0$, we have $f(\vec{g}_{t+1})\geq f(g_{t})$. Similarly, for any $q>\vec{h}_{t+1}(i)$, we have $f(\vec{e}_i|\vec{g}_t-\vec{g}_t(i)\vec{e}_i+(q-1)\vec{e}_i)<0$. Because of the dr-submodularity, we have $f(\vec{e}_i|\vec{h}_{t+1}-\vec{h}_{t+1}(i)\vec{e}_i+(q-1)\vec{e}_i)\leq  f(\vec{e}_i|\vec{g}_t-\vec{g}_t(i)\vec{e}_i+(q-1)\vec{e}_i)<0$, where $\vec{g}_{t}\preceq\vec{h}_{t+1}$. According to the Lemma 5, that is $f(\vec{h}_{t})=$
\begin{flalign}
&=f(\vec{h}_{t+1})+\sum_{i=1}^{d}\sum_{j=1}^{\vec{z}_t(i)}f\Bigg(\vec{e}_i|\vec{h}_{t+1}+\sum_{k=1}^{i-1}\sum_{l=1}^{\vec{z}_t(k)}\vec{e}_k+\sum_{l=1}^{j-1}\vec{e}_i\Bigg)\nonumber\\
&\leq f(\vec{h}_{t+1})+\sum_{i=1}^{d}\sum_{j=1}^{\vec{z}_t(i)}f\Bigg(\vec{e}_i|\vec{h}_{t+1}+\sum_{l=1}^{j-1}\vec{e}_i\Bigg)
\end{flalign}
where $\vec{z}_t(i)=\vec{h}_{t}(i)-\vec{h}_{t+1}(i)$. The inequality (12) is established since its dr-submodularity, that is $\vec{h}_{t+1}+\sum_{k=1}^{i-1}\sum_{l=1}^{\vec{z}_t(k)}\vec{e}_k\succeq\vec{h}_{t+1}$ definitely. Besides, since $f(\vec{e}_i|\vec{h}_{t+1}+\sum_{l=1}^{j-1}\vec{e}_i)<0$, we have $f(\vec{h}_{t+1})\geq f(\vec{h}_t)$.
\end{proof}
At this time, we can initialize $\vec{x},\vec{y}$ with $\vec{x}\leftarrow\vec{g}^\circ,\vec{y}\leftarrow\vec{h}^\circ$ instead of starting with $\vec{x}\leftarrow\vec{0},\vec{y}\leftarrow\vec{b}$ in Algorithm \ref{a1}, where the search space required to be checked is reduced to $[\vec{g}^\circ,\vec{h}^\circ]$. Then, we are able to build the approximation ratio for our revised lattice-based double greedy algorithm.
\begin{lem}
If we initialize Algorithm \ref{a1} by $[\vec{g}^\circ,\vec{h}^\circ]$, the solution $\vec{x}^\circ$ returned by Algorithm \ref{a1} satisfies
\begin{equation}
\mathbb{E}[f(\vec{x}^\circ)]\geq\frac{f((\vec{x}^*\lor\vec{g}^\circ)\land\vec{h}^\circ)+\frac{1}{2}\cdot(f(\vec{g}^\circ)+f(\vec{h}^\circ))}{2}
\end{equation}
\end{lem}
\begin{proof}
	It can be extended from the proof of Lemma 3.1 in \cite{buchbinder2015tight}. This procedure is complicated, so we omit here because of space limitation.
\end{proof}
\begin{thm}
	For our CPM-MS problem, if we initialize Algorithm \ref{a1} by $[\vec{g}^\circ,\vec{h}^\circ]$, and $f(\vec{g}^\circ)+f(\vec{h}^\circ)\geq 0$ is satisfied, the marketing vector $\vec{x}^\circ$ returned by Algorithm \ref{a1} is a $(1/2)$-approximate solution.
\end{thm}
\begin{proof}
	Based on Lemma 4, we have $\vec{g}^\circ\preceq\vec{x}^*\preceq\vec{h}^\circ$, hence $(\vec{x}^*\lor\vec{g}^\circ)\land\vec{h}^\circ=\vec{x}^*$. If $f(\vec{g}^\circ)+f(\vec{h}^\circ)\geq 0$, we can get that $\mathbb{E}[f(\vec{x}^\circ)]\geq(1/2)\cdot(f(\vec{x}^*)+(1/2)(f(\vec{g}^\circ)+f(\vec{h}^\circ)))\geq(1/2)\cdot f(\vec{x}^*)=(1/2)\cdot\max_{\vec{x}\preceq\vec{b}}f(\vec{x})$.
\end{proof}

From Theorem 5, it enables us to obtain the same approximation ratio by applying the lattice-based double greedy algorithm initialized by using iterative pruning if we have $f(\vec{g}^\circ)+f(\vec{h}^\circ)\geq 0$. According to the Lemma 6, that is $f(\vec{0})+f(\vec{b})=f(\vec{g}_0)+f(\vec{h}_0)\leq f(\vec{g}_1)+f(\vec{h}_1)\leq\cdots\leq f(\vec{g}^\circ)+f(\vec{h}^\circ)$. To achieve this condition $f(\vec{g}^\circ)+f(\vec{h}^\circ)\geq 0$ is much easier than $f(\vec{0})+f(\vec{b})\geq 0$. Therefore, the applications of Algorithm \ref{a1} with a theoretical bound are extended greatly by the technique of lattice-based iterative prunning.

\section{Speedup by Sampling Techniques}
In this section, we analyze the time complexity of Algorithm \ref{a1} and Algorithm \ref{a2}, and then discuss how to reduce their running time by sampling techniques.

\subsection{Time Complexity}
First, we assume there is a value oracle for computing the marginal gain of increasing or decreasing component $i\in[d]$ by $1$. If we initialize $\vec{x}\leftarrow\vec{0}$ and $\vec{y}\leftarrow\vec{b}$ at the beginning of lattice-based double greedy algorithm (Algorithm \ref{a1}), we have to take $2\cdot\sum_{i=1}^{d}\vec{b}(i)$ times together for checking each component whether to increase or decrease it by $1$. Consider shrinking collection $[\vec{0},\vec{b}]$ to $[\vec{g}^\circ,\vec{h}^\circ]$ by applying lattice-based iterative pruning (Algorithm \ref{a2}) first, we use it to initialize $\vec{x}$ and $\vec{y}$ at the beginning of Algorithm \ref{a1}, and then running the Algorithm \ref{a1}. For each component $i\in[d]$, we check its marginal gain $\vec{g}^\circ(i)+(\vec{b}(i)-\vec{h}^\circ(i))$ times in the iterative pruning, thus totally $\sum_{i=1}^{d}(\vec{g}^\circ(i)+(\vec{b}(i)-\vec{h}^\circ(i)))$ times. Then, we are required to check $2\cdot\sum_{i=1}^{d}(\vec{h}^\circ(i)-\vec{g}^\circ(i))$ times in subsequent double greedy initialized by $[\vec{g}^\circ,\vec{h}^\circ]$. Combine together, we have to check $\sum_{i=1}^{d}(\vec{b}(i)+\vec{h}^\circ(i)-\vec{g}^\circ(i))\leq2\cdot\sum_{i=1}^{d}\vec{b}(i)$ times. Hence, the time complexity is $O(\left\|\vec{b}\right\|_1)$. However, to compute the marginal gain of profit is a time consuming process, and the running time is given by Theorem 3, which is not acceptable in a large-scale social graph as well as a large searching space.

\subsection{Sampling Techniques}
To overcome the \#P-hardness of computing the objective $f(\cdot)$, we borrow from the idea of reverse influence sampling (RIS) \cite{borgs2014maximizing}. In the beginning, consider traditional IM problem, we need to introduce the concept of reverse reachable set (RR-set) first. Given a triggering model $\Omega=(G,\mathcal{D})$, a random RR-set can be generated by selecting a node $u\in V$ uniformly and sampling a graph realization $g$ from $\Omega$, then collecting those nodes can reach $u$ in $g$. RR-sets rooted
at $u$ is the collected nodes that are likely to influence $u$. A larger expected influence spread a seed set $S$ has, the higher the probability that $S$ intersects with a random RR-set is. Given a seed set $S$ and
a random RR-set $R$, we have $\sigma(S)=n\cdot\Pr[R\cap S\neq\emptyset]$. 

Extended to the lattice domain, given a marketing vector $\vec{x}$, its expected influence spread under the triggering model $\Omega$ can be denoted by $\mu(\vec{x})=n\cdot\mathbb{E}_{R}[1-\prod_{u\in R}(1-h_u(\vec{x}))]$ \cite{chen2018scalable}. Let
$\mathcal{R}=\{R_1,R_2,\cdots,R_\theta\}$ be a collection of random RR-sets generated independently, we have
\begin{equation}
	\hat{f}(\mathcal{R},\vec{x})=\frac{n}{\theta}\cdot\sum_{R\in\mathcal{R}}\left(1-\prod_{u\in R}(1-h_u(\vec{x}))\right)-c(\vec{x})
\end{equation}
that is an unbiased estimator of $f(\vec{x})$. From here, the vector $\vec{x}$ that maximizes $\hat{f}(\mathcal{R},\vec{x})$ will be close to the optimal solution intuitively, and more and more close with the increase of $|\mathcal{R}|$. Similar to $f(\vec{x})$, fix the collection $\mathcal{R}$, the $\hat{f}(\mathcal{R},\vec{x})$ is dr-submodular, but not monotone as well. By Theorem 5, Algorithm \ref{a1} offers a $(1/2)$-approximation if $\hat{f}(\mathcal{R},\vec{g}^\circ)+\hat{f}(\mathcal{R},\vec{h}^\circ)\geq 0$ is satisfied after the process of pruning.

Now, we begin to design our algorithm based on the idea of reverse sampling. First, we need to sample the enough number of random RR-sets so that its estimation to the objective function is accurate. Let $\varepsilon_1$, $\varepsilon_2$, and $\varepsilon_3$ be three adjustable parameters, where they satisfy
\begin{equation}
	\varepsilon_2+(1/2)\cdot\varepsilon_3=\varepsilon
\end{equation}
where $\varepsilon_1,\varepsilon_2,\varepsilon_3>0$. Then, we can set that
\begin{flalign}
&\theta_1=\sqrt{\frac{n^2\cdot\ln(3\delta\cdot\prod_{i=1}^{d}(\vec{b}(i)+1))}{2\varepsilon_1^2}}\\
&\theta_2=\frac{n(2n+\varepsilon_2^2\cdot\underline{OPT})\cdot\ln(3\delta\cdot\prod_{i=1}^{d}(\vec{b}(i)+1))}{\varepsilon_2^2\cdot\underline{OPT}^2}\\
&\theta_3=\frac{2n^2\cdot\ln(3\delta)}{\varepsilon_3^2\cdot\underline{OPT}^2}
\end{flalign}
where the $\underline{OPT}$ is the lower bound of the optimal objective $f(\vec{x}^*)$. The algorithm that combining double greedy with reverse sampling and iterative pruning, called DG-IP-RIS algorithm, is shown in Algorithm \ref{a3}.

\begin{algorithm}[!t]
	\caption{\text{DG-IP-RIS}}\label{a3}
	\begin{algorithmic}[1]
		\renewcommand{\algorithmicrequire}{\textbf{Input:}}
		\renewcommand{\algorithmicensure}{\textbf{Output:}}
		\REQUIRE $\hat{f}:\mathcal{R}\times\mathbb{Z}^d_+\rightarrow\mathbb{R}$, $\vec{b}\in\mathbb{Z}^d_+$, $(\varepsilon_1,\varepsilon_2,\varepsilon_3)$, $\delta$
		\ENSURE $\hat{\vec{x}}\in\mathbb{Z}^d_+$
		\STATE Initialize: $\theta_1$ define on (16)
		\STATE $\underline{OPT}\leftarrow$ OptEstimation$(\hat{f},\vec{b},\theta_1)$
		\STATE Initialize: $\theta_2,\theta_3$ defined on (17) (18)
		\STATE $\theta\leftarrow\max(\theta_1,\theta_2,\theta_3)$
		\STATE Generate a collection of random RR-sets $\mathcal{R}$ with $|\mathcal{R}|=\theta$
		\STATE $[\hat{\vec{g}}^\circ,\hat{\vec{h}}^\circ]\leftarrow$ Lattice-basedPruning$(\hat{f}(\mathcal{R},\cdot),\vec{b})$
		\STATE $\hat{\vec{x}}\leftarrow$ Lattice-basedDoubleGreedy$(\hat{f}(\mathcal{R},\cdot),[\hat{\vec{g}}^\circ,\hat{\vec{h}}^\circ])$
		\RETURN $\hat{\vec{x}}$
	\end{algorithmic}
\end{algorithm}

In DG-IP-RIS algorithm, we estimate the number of random RR-sets $\theta$ in line 4, and then generate a collection $\mathcal{R}$ of random RR-sets with the size of $\theta$. The objective $\hat{f}(\mathcal{R},\cdot)$ is computed based on this $\mathcal{R}$, from which we are able to get a solution by iterative pruning and double greedy algorithm. In the first step, we require to compute a lower bound of optimal value $f(\vec{x}^*)$, which is shown in Algorithm \ref{a4}. Here, we increase the component by $1$ with the largest marginal gain at each iteration until there is no component having positive marginal gain. After the while loop, we can obtain a vector $\vec{x}$ and set $\underline{OPT}\leftarrow\hat{f}(\mathcal{R},\vec{x})-2\varepsilon_1$, because $\Pr[|\hat{f}(\mathcal{R},\vec{x})-f(\vec{x})|\leq\varepsilon_1]$ is satisfied with a high probability under the setting of $\theta_1$. Like this, we have $\underline{OPT}>0$ as well because the largest marginal gain should be greater than $0$ due to the dr-submodularity, or else the definition of our problem is not valid and meaningless. For convenience, given a random RR-set $R$, we denote $p(R,\vec{x})=1-\prod_{u\in R}(1-h_u(\vec{x}))$ in subsequent proof.

\begin{algorithm}[!t]
	\caption{\text{OptEstimation}}\label{a4}
	\begin{algorithmic}[1]
		\renewcommand{\algorithmicrequire}{\textbf{Input:}}
		\renewcommand{\algorithmicensure}{\textbf{Output:}}
		\REQUIRE $\hat{f}:\mathcal{R}\times\mathbb{Z}^d_+\rightarrow\mathbb{R}$, $\vec{b}\in\mathbb{Z}^d_+$, $\theta_1$
		\ENSURE $\underline{OPT}$
		\STATE Initialize: $\vec{x}\leftarrow0$, $t\leftarrow0$
		\STATE Generate a collection of random RR-sets $\mathcal{R}$ with $|\mathcal{R}|=\theta_1$
		\WHILE {$t<\sum_{i=0}^{d}\vec{b}(i)$}
		\STATE $i^*\leftarrow\arg\max_{i\in[d],\vec{x}(i)<\vec{b}(i)}\hat{f}(\vec{e}_i|\mathcal{R},\vec{x})$
		\IF {$\hat{f}(\vec{e}_{i^*}|\mathcal{R},\vec{x})\leq 0$}
		\STATE Break
		\ENDIF
		\STATE $\vec{x}\leftarrow\vec{x}+\vec{e}_{i^*}$, $t\leftarrow t+1$
		\ENDWHILE
		\STATE $\underline{OPT}\leftarrow\hat{f}(\mathcal{R},\vec{x})-2\varepsilon_1$
		\RETURN $\underline{OPT}$
	\end{algorithmic}
\end{algorithm}

\begin{lem}
	The $\underline{OPT}$ returned by Algorithm \ref{a4} satisfies $f(\vec{x}^*)\geq\underline{OPT}$ with at least $1-1/(3\delta)$ probability.
\end{lem}
\begin{proof}
	For any marketing vector $\vec{x}$, we want to obtain $\Pr[|\hat{f}(\mathcal{R},\vec{x})-f(\vec{x})|\geq\varepsilon_1]\leq1/(3\delta\cdot\prod_{i=1}^{d}\vec{b}(i))$. By the additive form of Chernoff-Hoeffding Inequality, it is equivalent to compute, that is
	\begin{equation*}
	\Pr\left[\left|\frac{1}{\theta_1}\sum p(R_i,\vec{x})-\frac{\mu(\vec{x})}{n}\right|\geq\frac{\varepsilon_1}{n}\right]\leq\exp\left(-\frac{2\theta_1^2\varepsilon_1^2}{n^2}\right)
	\end{equation*}
	When $\theta_1$ is defined as (16), we have $1/(3\delta\cdot\prod_{i=1}^{d}(\vec{b}(i)+1))=\exp(-2\theta_1^2\varepsilon_1^2/n^2)$ definitely. By the union bound, the above relationship holds for the $\vec{x}'$ generated in line 10 of Algorithm \ref{a3} with a probablity less than $1/(3\delta)$.
\end{proof}
\begin{rem}
	Given a marketing vector $\vec{x}\preceq\vec{b}$, for each component $i\in[d]$, the possible values of $\vec{x}(i)$ are $\{0,1,2,\cdots,\vec{b}(i)\}$, thus the number of possible values for $\vec{x}(i)$ is $\vec{b}(i)+1$. Thereby the total number of possible combinations for vector $\vec{x}$ is $\prod_{i=1}^{d}(\vec{b}(i)+1)$, which explains why the union bound in the previous lemma happened.
\end{rem}

\begin{lem}[Chernoff bounds \cite{motwani1995randomized}]
	Given a collection $\mathcal{Z}=\{Z_1,Z_2,\cdots,Z_\theta\}$, each $Z_i\in[0,1]$ is an i.i.d. random variable with $\mathbb{E}[Z_i]=\nu$, we have
	\begin{equation}
	\Pr\left[\sum\nolimits_{i=1}^\theta Z_i\geq(1+\gamma)\cdot\nu\theta\right]\leq\exp\left(-\frac{\gamma^2\cdot\nu\theta}{2+\gamma}\right)
	\end{equation}
	\begin{equation}
	\Pr\left[\sum\nolimits_{i=1}^\theta Z_i\leq(1-\gamma)\cdot\nu\theta\right]\leq\exp\left(-\frac{\gamma^2\cdot\nu\theta}{2}\right)
	\end{equation}
	where we assume that $\gamma>0$.
\end{lem}

\begin{lem}
	Given a collection $\mathcal{R}$ with $|\mathcal{R}|=\theta_2$, for any marketing vector $\vec{x}\preceq\vec{b}$, it satisfies $\hat{f}(\mathcal{R},\vec{x})-f(\vec{x})<\varepsilon_2\cdot f(\vec{x}^*)$ with at least $1-1/(3\delta)$ probability.
\end{lem}
\begin{proof}
	For any marketing vector $\vec{x}$, we want to obtain $\Pr[\hat{f}(\mathcal{R},\vec{x})-f(\vec{x})\geq\varepsilon_2\cdot f(\vec{x}^*)]\leq1/(3\delta\cdot\prod_{i=1}^{d}(\vec{b}(i)+1))$. By the Chernoff bound, defined as (19), it is equivalent to compute, that is
	\begin{flalign}
	&\Pr\left[\sum p(R_i,\vec{x})\geq\left(1+\frac{\varepsilon_2f(\vec{x}^*)}{\mu(\vec{x})}\right)\cdot\frac{\mu(\vec{x})\theta_2}{n}\right]\nonumber\\
	&\leq\exp\left(-\frac{\left(\frac{\varepsilon_2f(\vec{x}^*)}{\mu(\vec{x})}\right)^2\cdot\frac{\mu(\vec{x})\theta_2}{n}}{2+\frac{\varepsilon_2f(\vec{x}^*)}{\mu(\vec{x})}}\right)
	\end{flalign}
	From Lemma 9 and $\mu(\vec{x})\leq n$, we have
	\begin{equation*}
		(21)\leq\exp\left(-\frac{\theta_2\cdot\varepsilon_2^2\cdot\underline{OPT}^2}{n\cdot(2n+\varepsilon_2\cdot\underline{OPT})}\right)\leq\frac{1}{3\delta\cdot\prod_{i=1}^{d}(\vec{b}(i)+1)}
	\end{equation*}
	By the union bound, the above relationship holds for any $\vec{x}\preceq\vec{b}$ with at most $1/(3\delta)$ probability.
\end{proof}

\begin{lem}
	Given a collection $\mathcal{R}$ with $|\mathcal{R}|=\theta_3$, for an optimal solution $\vec{x}^*$, it satisfies $\hat{f}(\mathcal{R},\vec{x}^*)-f(\vec{x}^*)>-\varepsilon_3\cdot f(\vec{x}^*)$ with at least $1-1/(3\delta)$ probability.
\end{lem}
\begin{proof}
	For an optimal solution $\vec{x}^*$, we want to obtain $\Pr[\hat{f}(\mathcal{R},\vec{x}^*)-f(\vec{x}^*)\leq-\varepsilon_3\cdot f(\vec{x}^*)]\leq1/(3\delta)$. By the Chernoff bound, defined as (20), it is equivalent to compute, that is
	\begin{flalign}
		&\Pr\left[\sum p(R_i,\vec{x}^*)\leq\left(1-\frac{\varepsilon_3f(\vec{x}^*)}{\mu(\vec{x}^*)}\right)\cdot\frac{\mu(\vec{x}^*)\theta_2}{n}\right]\nonumber\\
		&\leq\exp\left(-\frac{\left(\frac{\varepsilon_3f(\vec{x}^*)}{\mu(\vec{x}^*)}\right)^2\cdot\frac{\mu(\vec{x}^*)\theta_3}{n}}{2}\right)
	\end{flalign}
	From Lemma 9 and $\mu(\vec{x})\leq n$, we have
	\begin{equation*}
		(22)\leq\exp\left(-\frac{\theta_3\cdot\varepsilon_3^2\cdot\underline{OPT}^2}{2n^2}\right)\leq\frac{1}{3\delta}
	\end{equation*}
	The above relationship holds for the optimal solution $\vec{x}^*$ with at most $\leq1/(3\delta)$ probability.
\end{proof}

Let $\hat{\vec{x}}^\circ$ be the result returned by Algorithm \ref{a3}. If $\hat{f}(\mathcal{R},\hat{\vec{x}}^\circ)$ and $\hat{f}(\mathcal{R},\vec{x}^*)$ are accurate estimations to the $f(\hat{\vec{x}}^\circ)$ and $f(\vec{x}^*)$, we can say this solution $\hat{\vec{x}}$ has an effective approximation guarantee, which is shown in Theorem 6.

\begin{thm}
	For our CPM-MS problem, if it satisfies $\hat{f}(\mathcal{R},\hat{\vec{g}}^\circ)+\hat{f}(\mathcal{R},\hat{\vec{h}}^\circ)\geq 0$, for any $\varepsilon\in(0,1/2)$ and $\delta>0$, the marketing vector $\hat{\vec{x}}^\circ$ returned by Algorithm \ref{a3} is a $(1/2-\varepsilon)$-approximation, that is
	\begin{equation}
		f(\hat{\vec{x}}^\circ)\geq(1/2-\varepsilon)\cdot f(\vec{x}^*)
	\end{equation}
	holds with at least $1-1/\delta$ probability.
\end{thm}
\begin{proof}
	Based on Lemma 10, $\hat{f}(\mathcal{R},\hat{\vec{x}}^\circ)-f(\hat{\vec{x}}^\circ)<\varepsilon_2\cdot f(\vec{x}^*)$ holds with at least $1-1/(3\delta)$ probability, and on Theorem 5, we have $\hat{f}(\mathcal{R},\hat{\vec{x}}^\circ)\geq(1/2)\cdot\hat{f}(\mathcal{R},\vec{x}^*)$. Thus, $f(\hat{\vec{x}}^\circ)\geq\hat{f}(\mathcal{R},\hat{\vec{x}}^\circ)-\varepsilon_2\cdot f(\vec{x}^*)\geq(1/2)\cdot\hat{f}(\mathcal{R},\vec{x}^*)-\varepsilon_2\cdot f(\vec{x}^*)$. By Lemma 11, $\hat{f}(\mathcal{R},\vec{x}^*)-f(\vec{x}^*)>-\varepsilon_3\cdot f(\vec{x}^*)$ holds with at least $1-1/(3\delta)$ probability, thus we have $f(\hat{\vec{x}}^\circ)\geq(1/2-(\varepsilon_2+1/2\cdot\varepsilon_3))\cdot f(\vec{x}^*)=(1/2-\varepsilon)\cdot f(\vec{x}^*)$. Combined with that $f(\vec{x}^*)\geq\underline{OPT}$ holds with $1-1/(3\delta)$, by the union bound, (23) holds with at least $1-1/\delta$ probability.
\end{proof}
Finally, we consider the running time of Algorithm \ref{a3}. Given a collection $R$ with $|R|=\theta=\max\{\theta_1,\theta_2,\theta_3\}$, we have $\theta=O(n^2)$. To compute $\hat{f}(R,\cdot)$, it takes $O(\theta n)$ time, and to generate a random RR-set, it takes $O(m)$ times. Thus, the time complexity of Algorithm \ref{a3} is $O(m\theta+\left\|\vec{b}\right\|_1\cdot n\theta)=O((m+n)n^2)$. Besides, this running time can be reduced further. Look at the forms of (16) (17) and (18), $\theta_1$ is apparently less than $\theta_2$ and $\theta_3$. Therefore, we are able to select the remaining two parameters $(\varepsilon_2,\varepsilon_3)$ such that $(\varepsilon_2,\varepsilon_3)=\arg\min_{\varepsilon_2+(1/2)\cdot\varepsilon_3=\varepsilon}\max\{\theta_2,\theta_3\}$.

\section{Experiments}
In this section, we carry out several experiments on different datasets to validate the performance of our proposed algorithms. It aims to test the efficiency of DG-IP-RIS algorithm (Algorithm \ref{a3}) and its effectiveness compared to other heuristic algorithms. All of our experiments are programmed by python, and run on Windows machine with a 3.40GHz, 4 core Intel CPU and 16GB RAM. There are four datasets used in our experiments: (1) NetScience \cite{nr}: a co-authorship network, co-authorship among scientists to publish papers about network science; (2) Wiki \cite{nr}: a who-votes-on-whom network, which come from the collection Wikipedia voting; (3) HetHEPT \cite{snapnets}: an academic collaboration relationship on high energy physics area; (4) Epinions \cite{snapnets}: a who-trust-whom online social network on Epinions.com, a general consumer review site. The statistics information of these four datasets is represented in Table \ref{table1}. For undirected graph, each undirected edge is replaced with two reversed directed edges.

\begin{table}[h]
	\renewcommand{\arraystretch}{1.3}
	\caption{The datasets statistics $(K=10^3)$}
	\label{table1}
	\centering
	\begin{tabular}{|c|c|c|c|c|}
		\hline
		\bfseries Dataset & \bfseries n & \bfseries m & \bfseries Type & \bfseries Avg.Degree\\
		\hline
		NetScience & 0.4K & 1.01K & undirected & 5.00\\
		\hline
		Wiki & 1.0K & 3.15K & directed & 6.20\\
		\hline
		HetHEPT & 12.0K & 118.5K & undirected & 19.8\\
		\hline
		Epinions & 75.9K & 508.8K & directed & 13.4\\
		\hline
	\end{tabular}
\end{table}

\begin{figure}[!t]
	\centering
	\subfigure[NetScience, Performance]{
		\includegraphics[width=0.48\columnwidth]{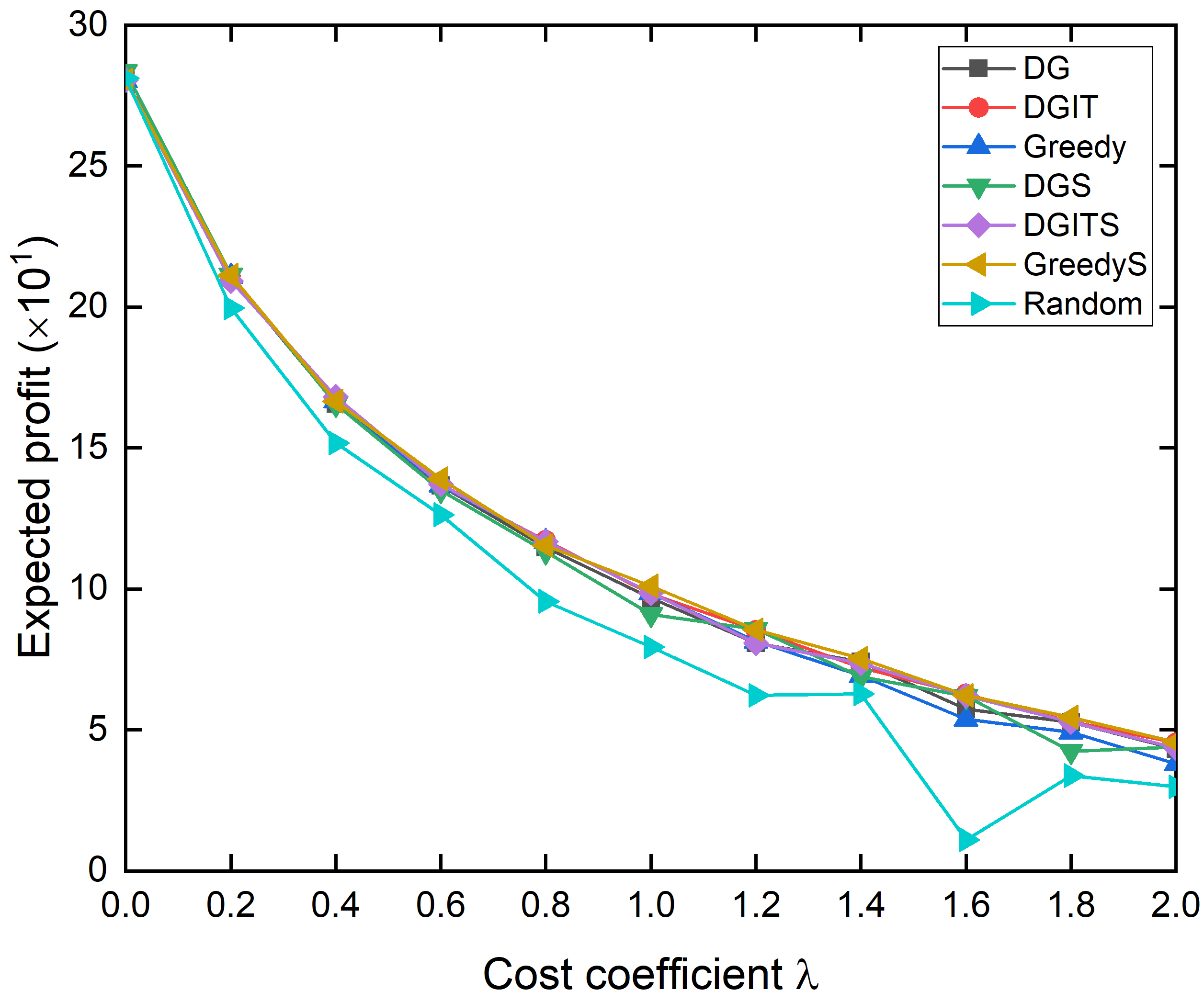}
		%\caption{fig1}
	}%
	\subfigure[NetScience, Time (s)]{
		\centering
		\includegraphics[width=0.48\columnwidth]{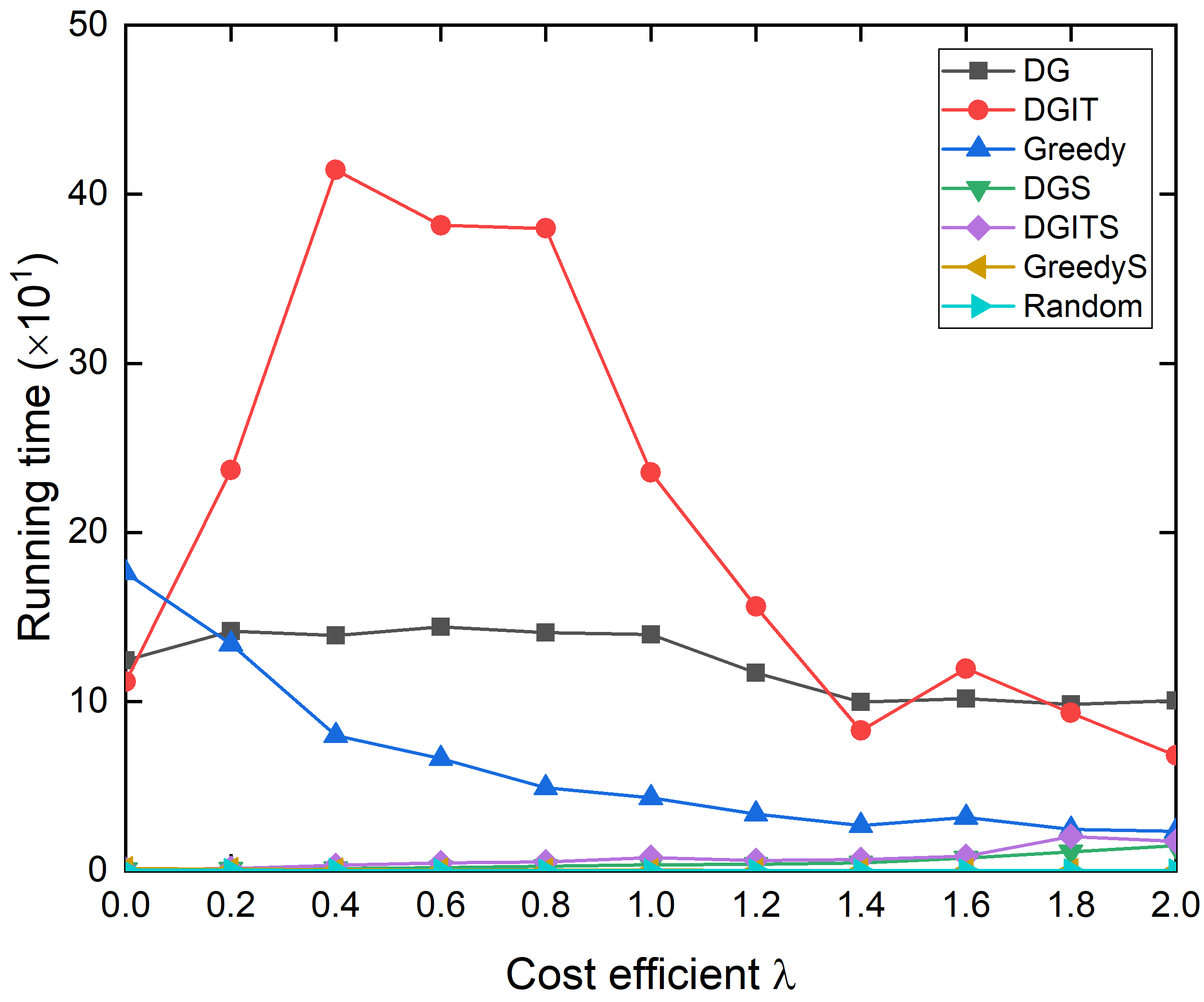}
		%\caption{fig2}
	}%
	
	\subfigure[Wiki, Performance]{
		\centering
		\includegraphics[width=0.48\columnwidth]{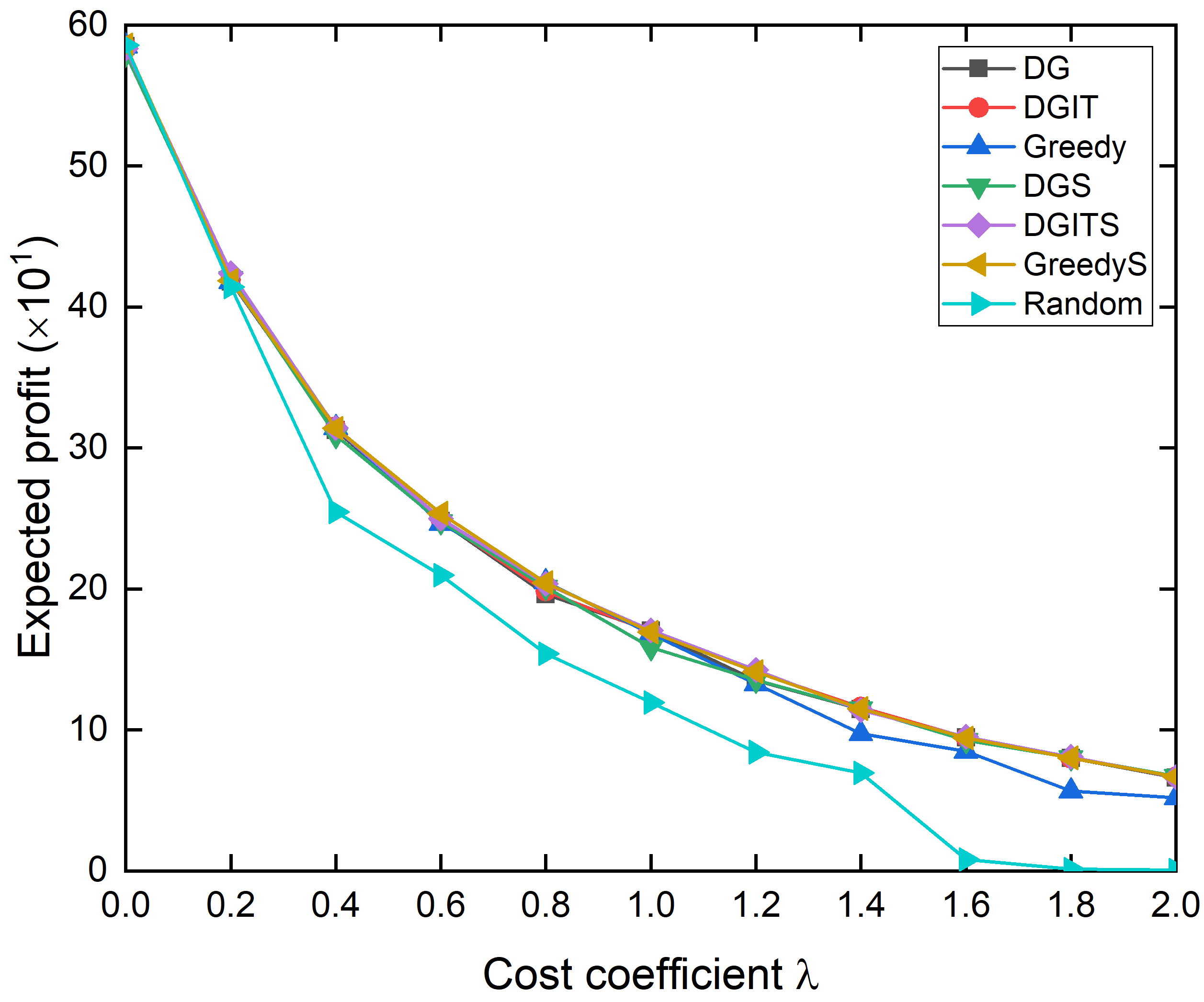}
		%\caption{fig2}
	}%
	\subfigure[Wiki, Time (s)]{
		\centering
		\includegraphics[width=0.48\columnwidth]{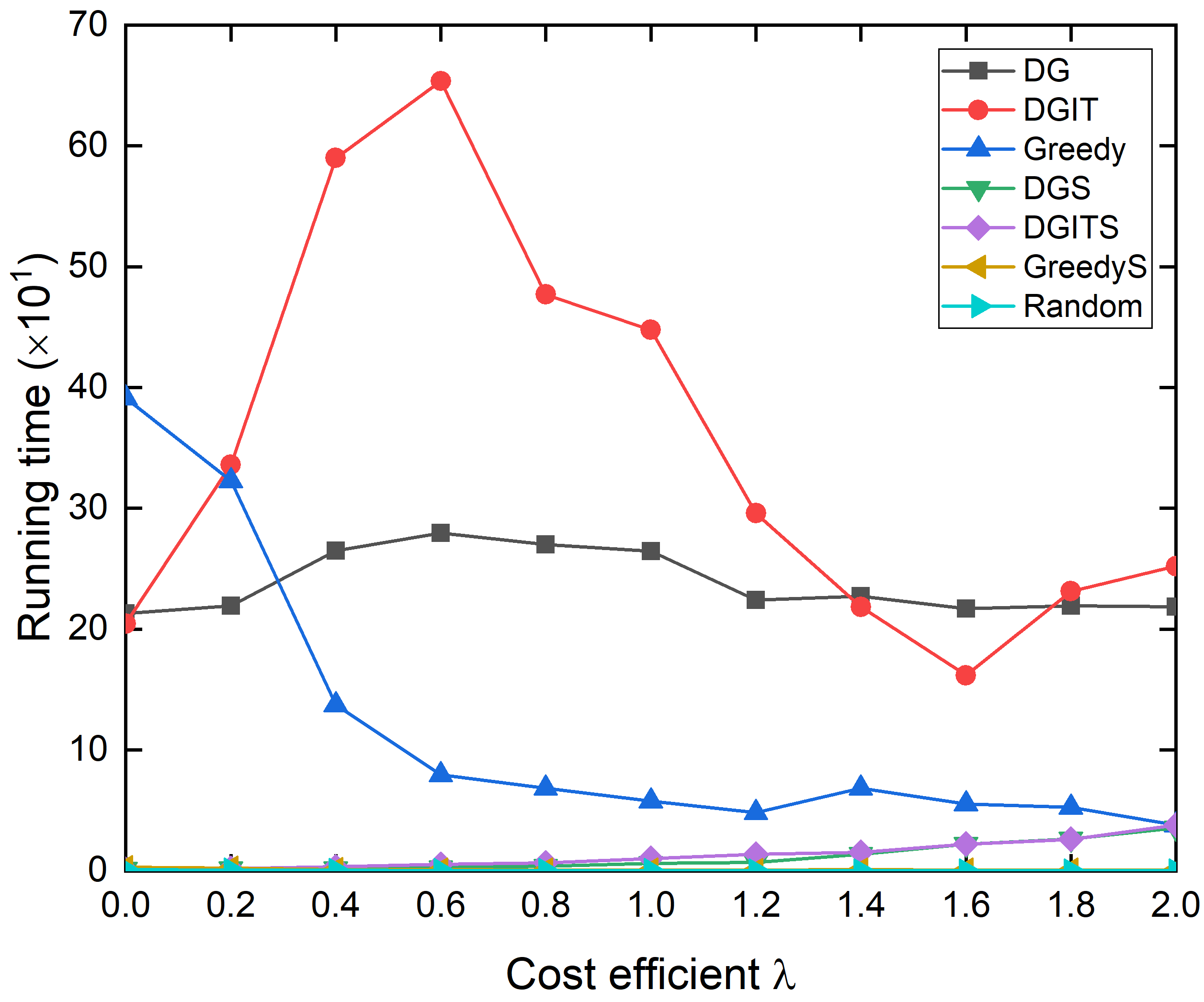}
		%\caption{fig2}
	}%
	
	\subfigure[HetHEPT, Performance]{
		\centering
		\includegraphics[width=0.48\columnwidth]{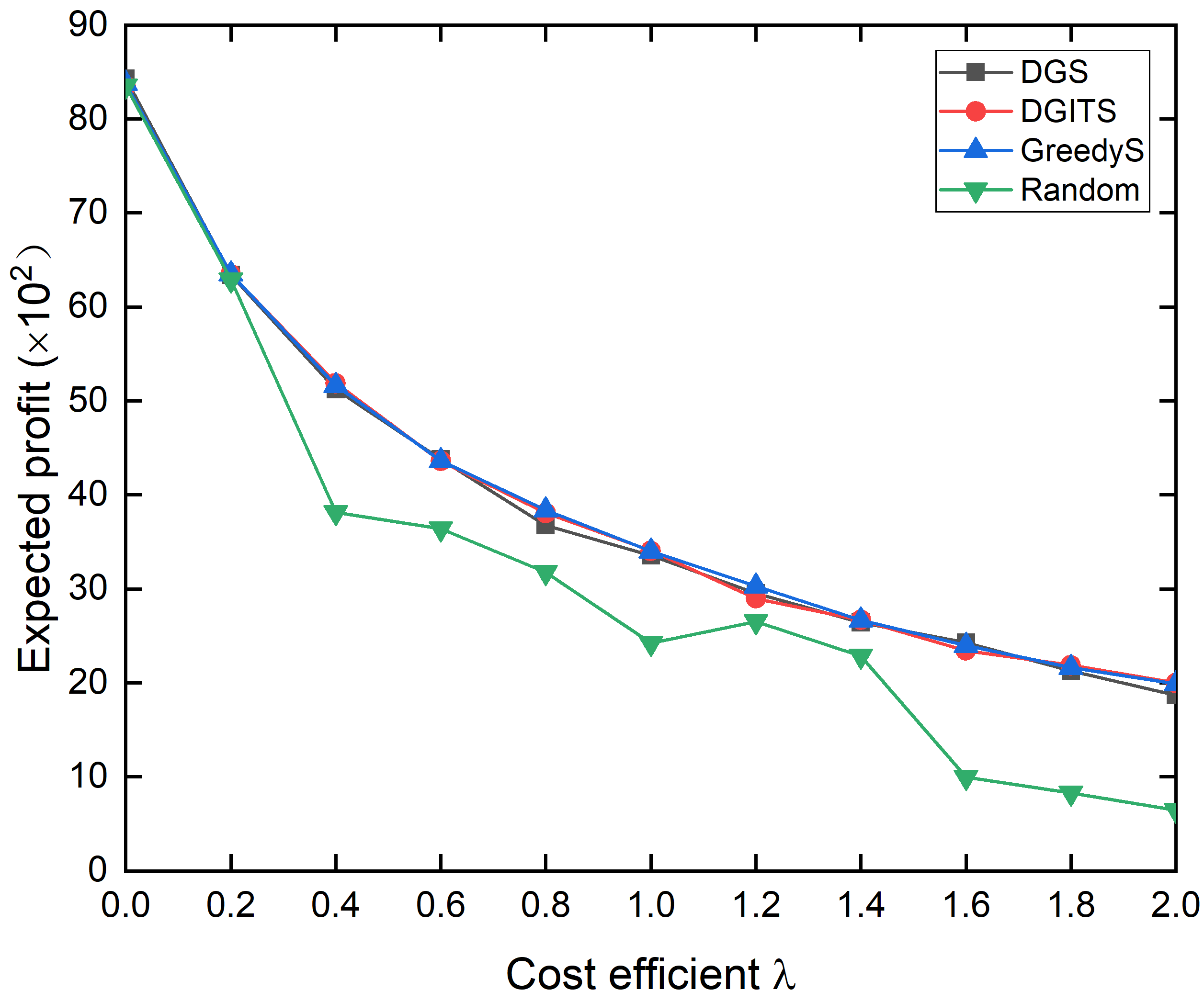}
		%\caption{fig2}
	}%
	\subfigure[HetHEPT, Time (s)]{
		\centering
		\includegraphics[width=0.48\columnwidth]{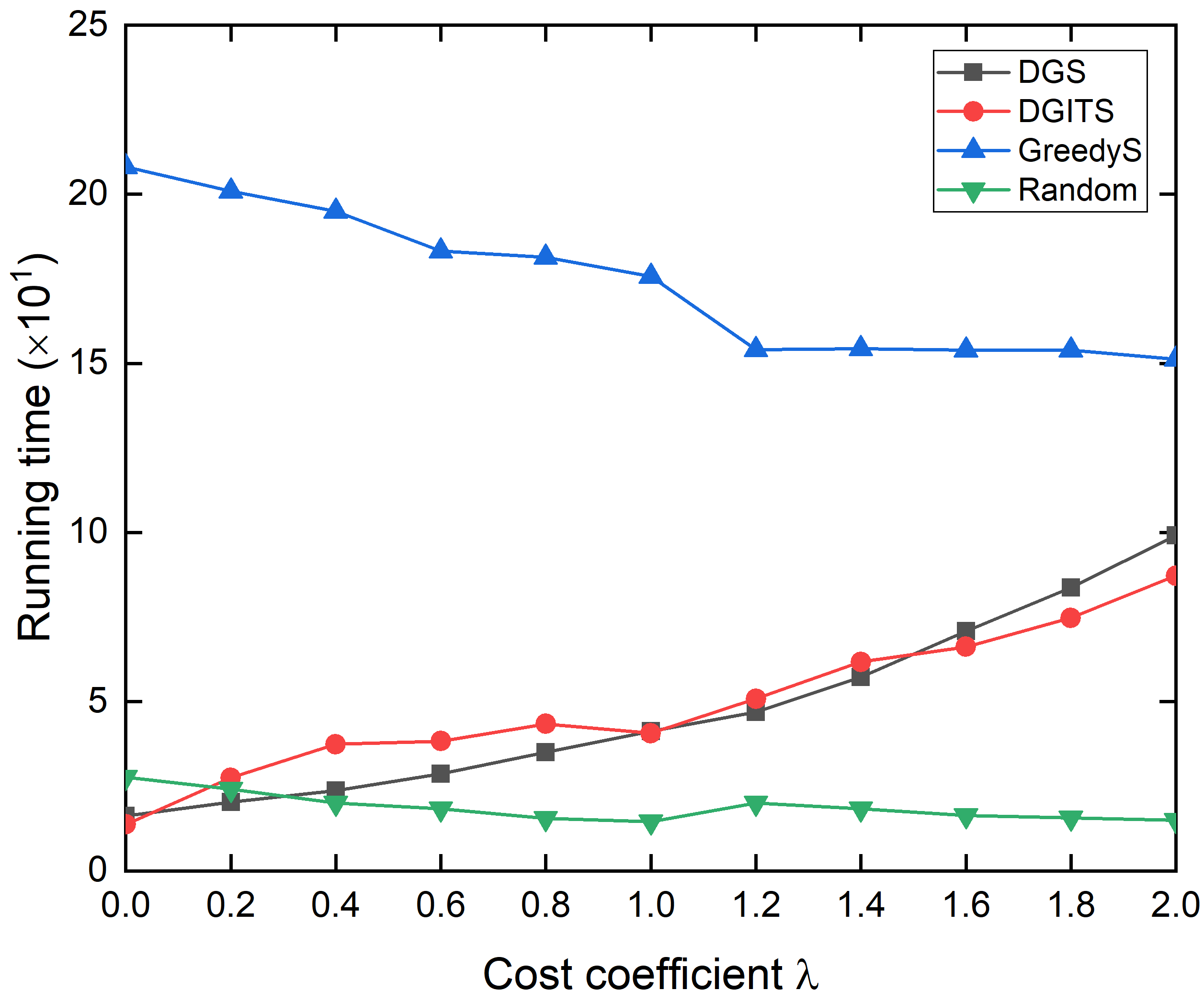}
		%\caption{fig2}
	}%
	
	\subfigure[Epinions, Performance]{
		\centering
		\includegraphics[width=0.48\columnwidth]{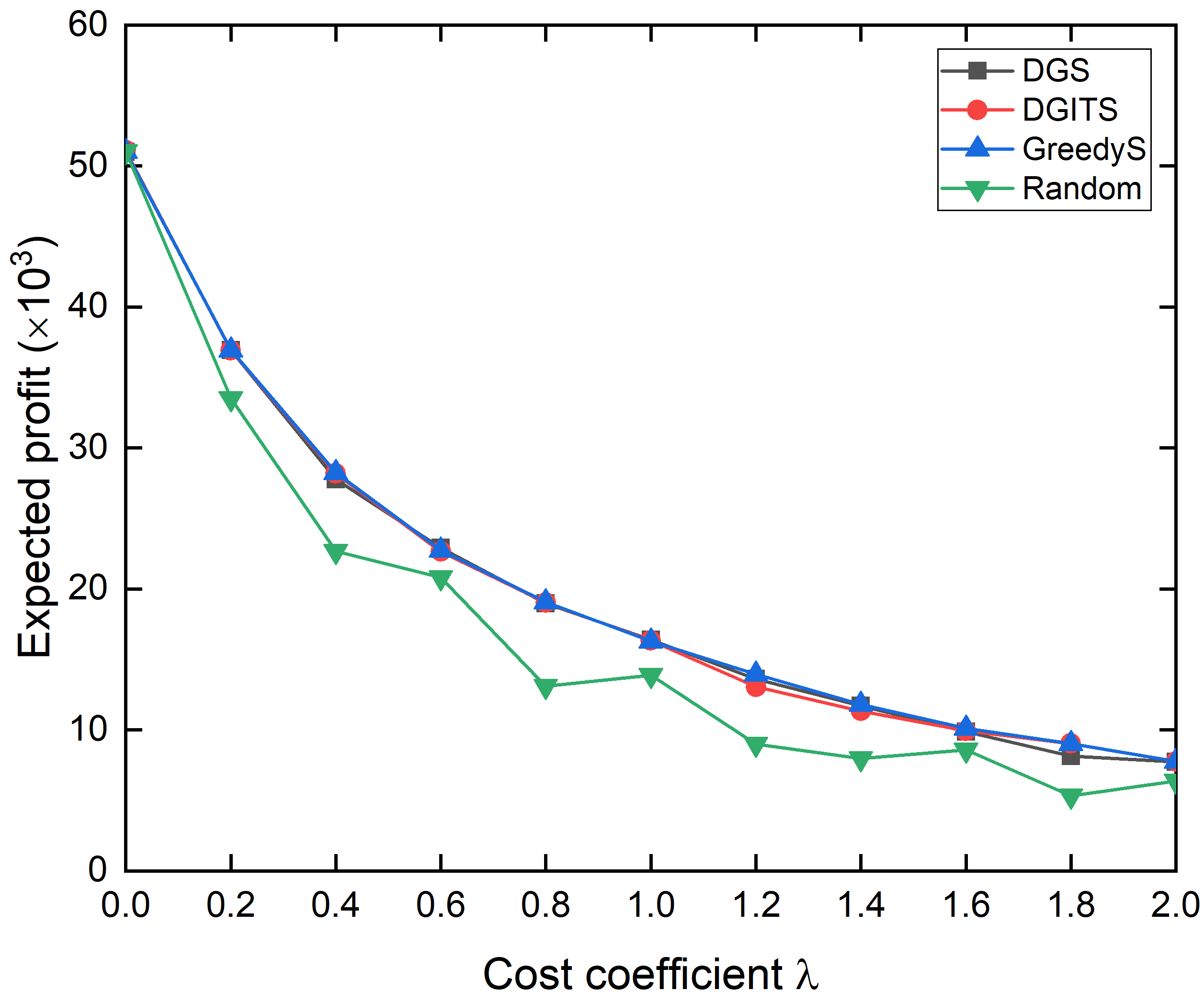}
		%\caption{fig2}
	}%
	\subfigure[Epinions, Time (s)]{
		\centering
		\includegraphics[width=0.48\columnwidth]{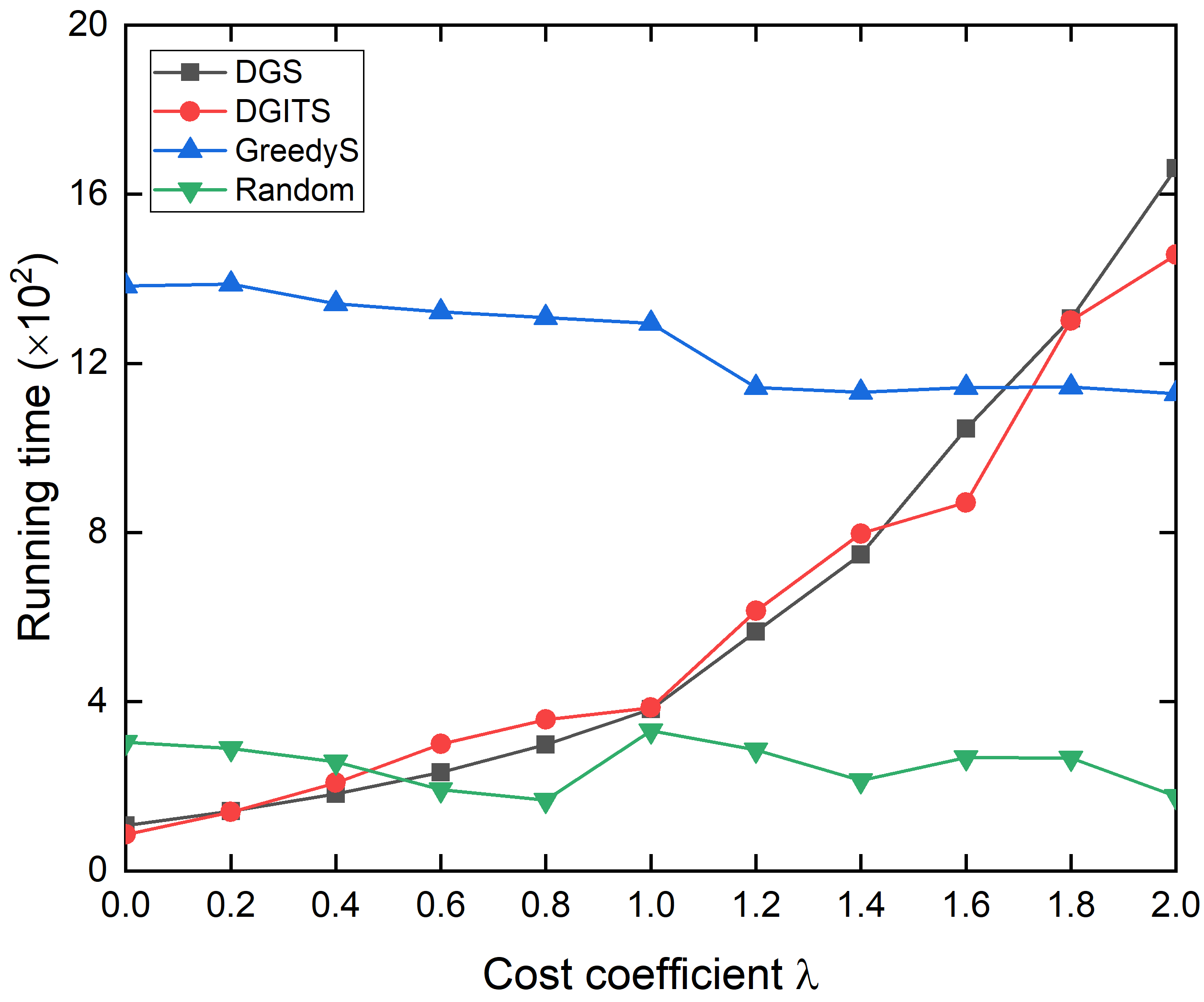}
		%\caption{fig2}
	}%
	\centering
	\caption{The performance and running time comparisons among different algorithms under the IC-model.}
	\label{fig1}
\end{figure}

\begin{figure}[!t]
	\centering
	\subfigure[NetScience, Performance]{
		\includegraphics[width=0.48\columnwidth]{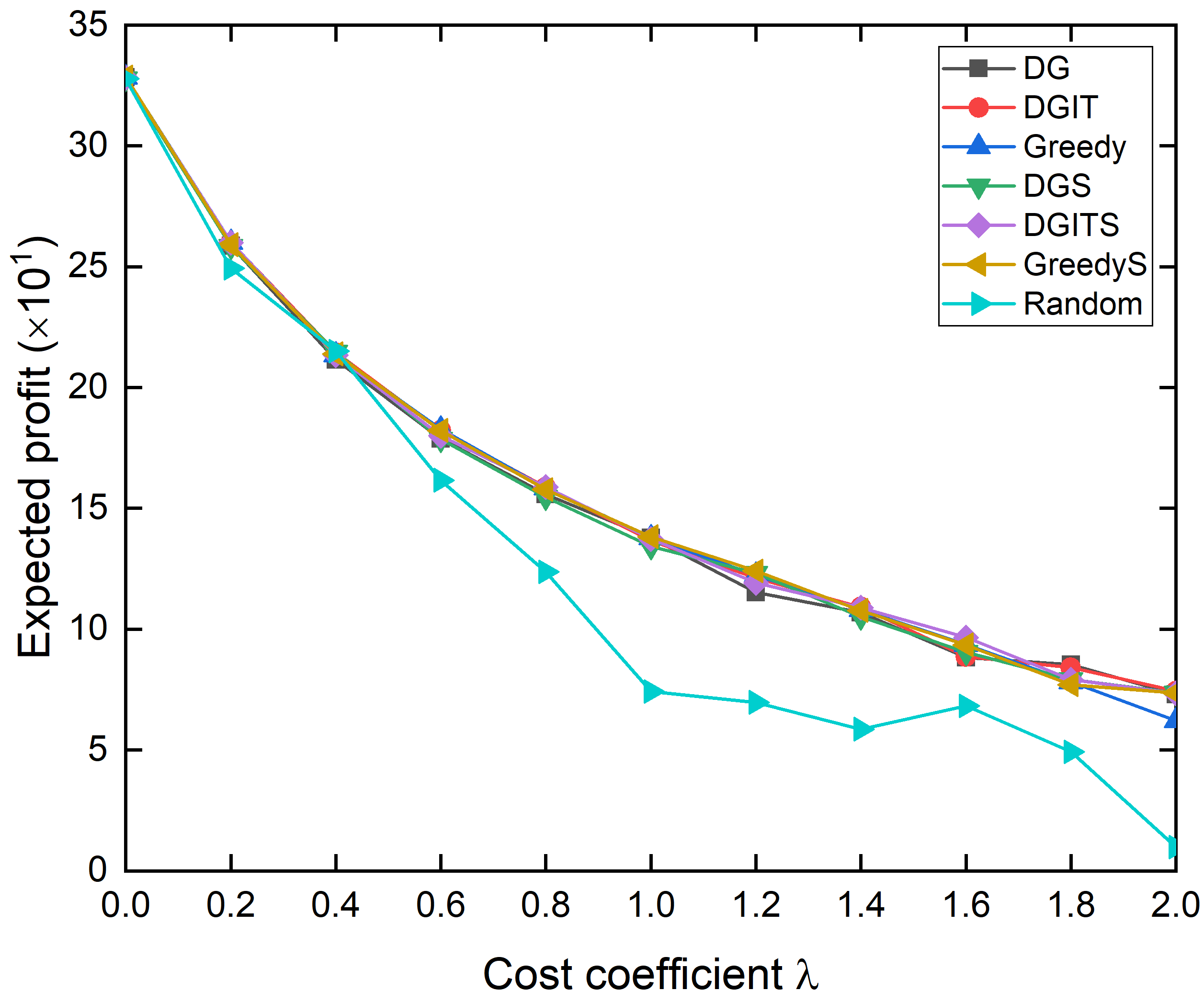}
		%\caption{fig1}
	}%
	\subfigure[NetScience, Time (s)]{
		\centering
		\includegraphics[width=0.48\columnwidth]{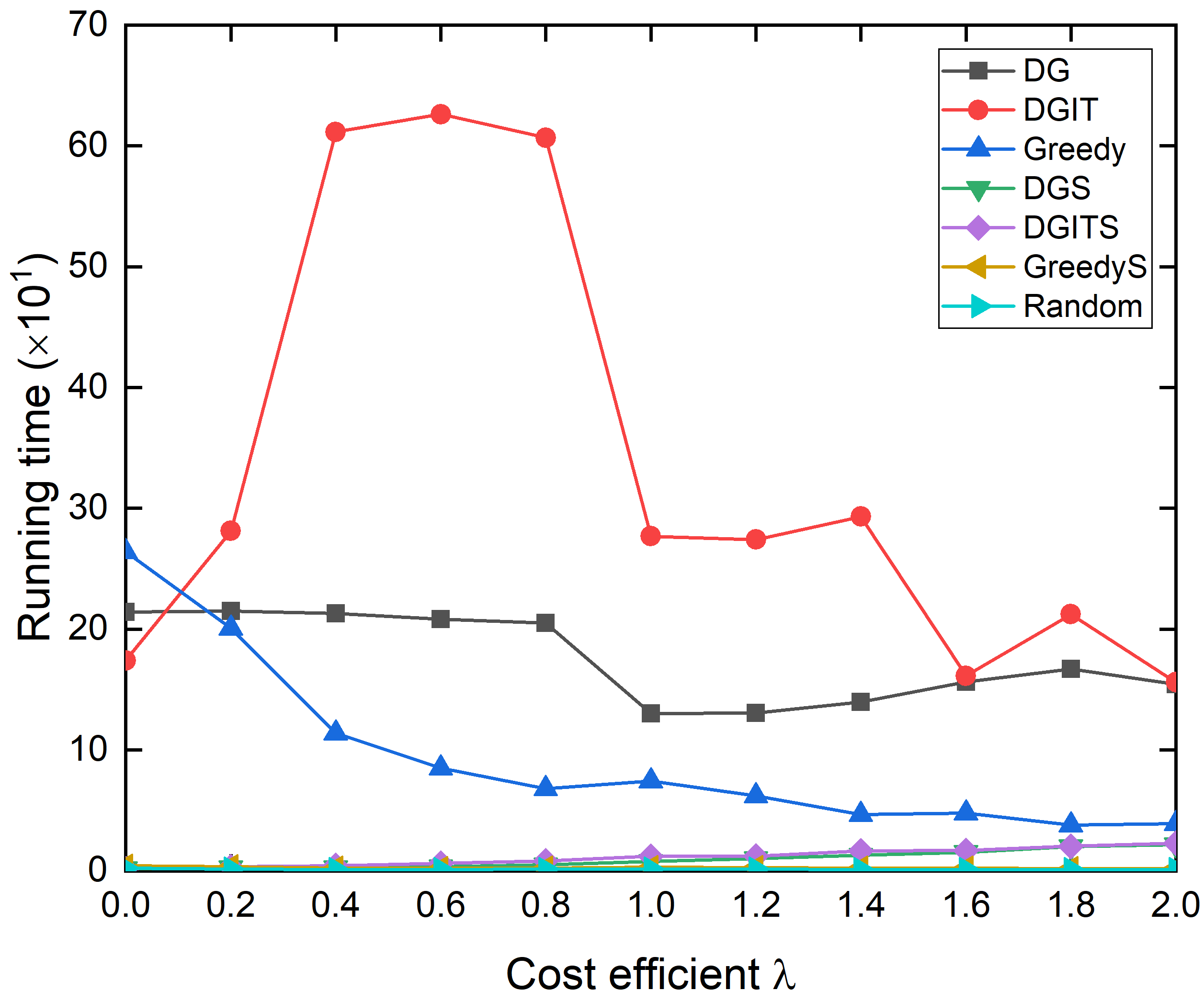}
		%\caption{fig2}
	}%
	
	\subfigure[Wiki, Performance]{
		\centering
		\includegraphics[width=0.48\columnwidth]{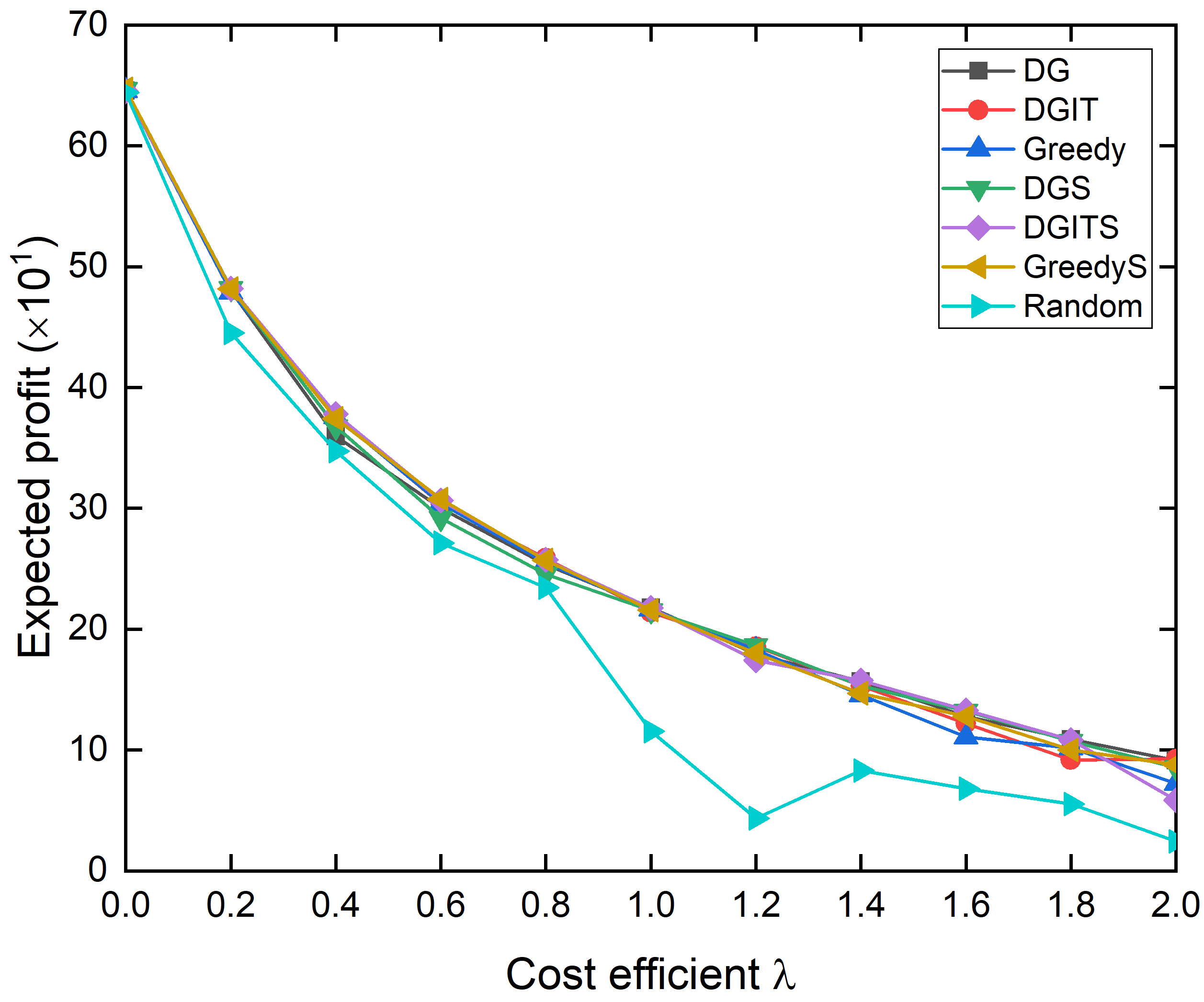}
		%\caption{fig2}
	}%
	\subfigure[Wiki, Time (s)]{
		\centering
		\includegraphics[width=0.48\columnwidth]{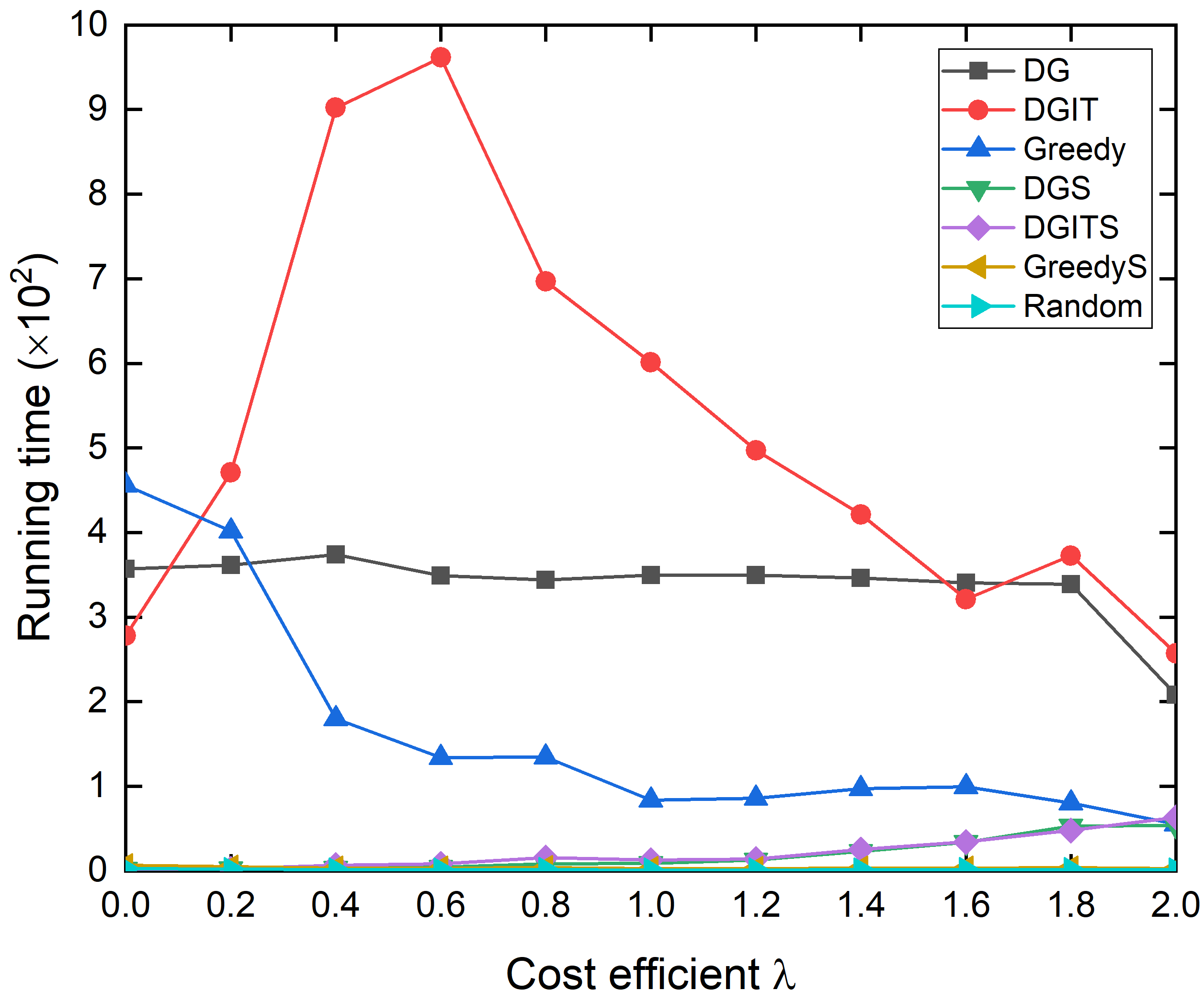}
		%\caption{fig2}
	}%
	
	\subfigure[HetHEPT, Performance]{
		\centering
		\includegraphics[width=0.48\columnwidth]{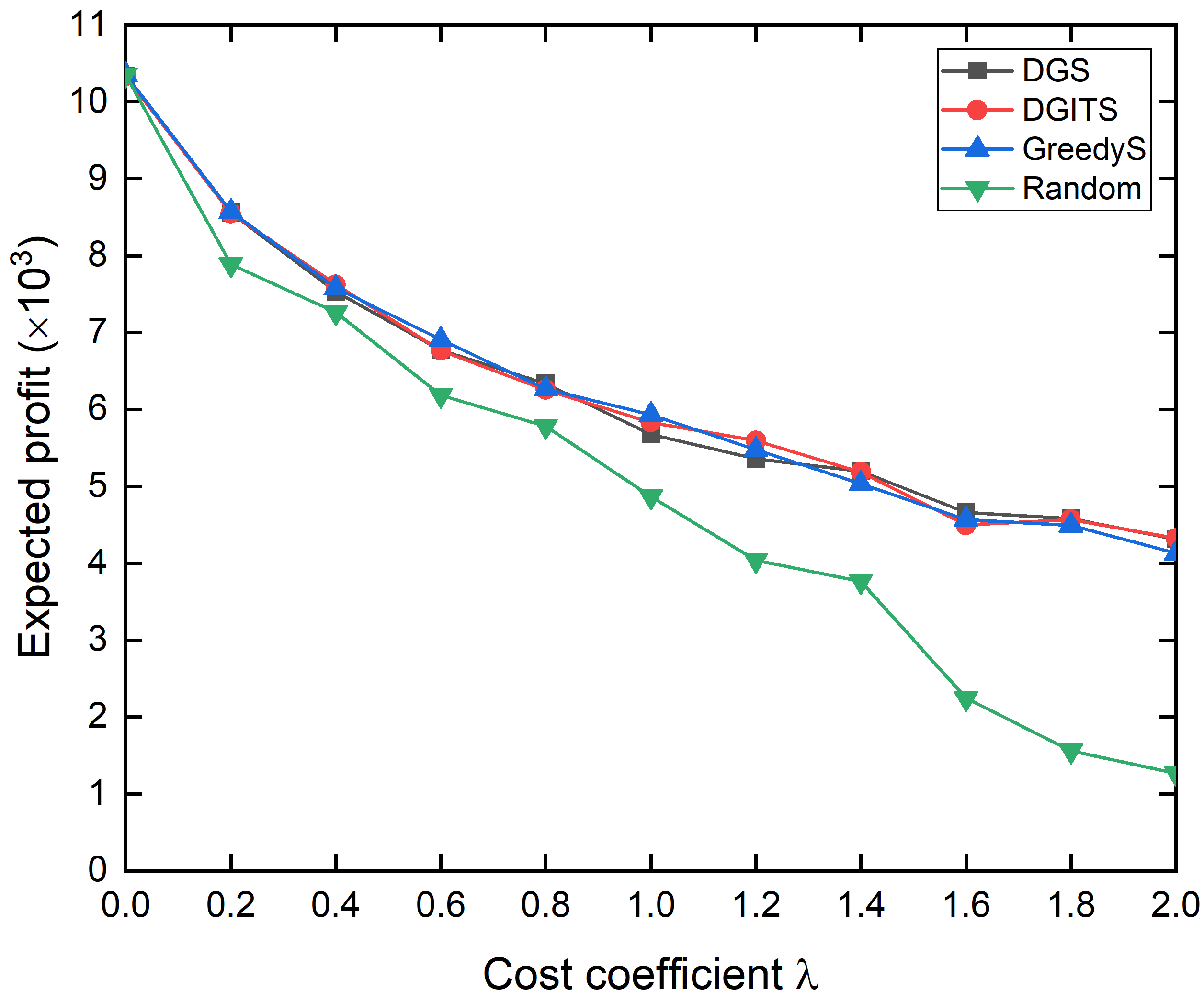}
		%\caption{fig2}
	}%
	\subfigure[HetHEPT, Time (s)]{
		\centering
		\includegraphics[width=0.48\columnwidth]{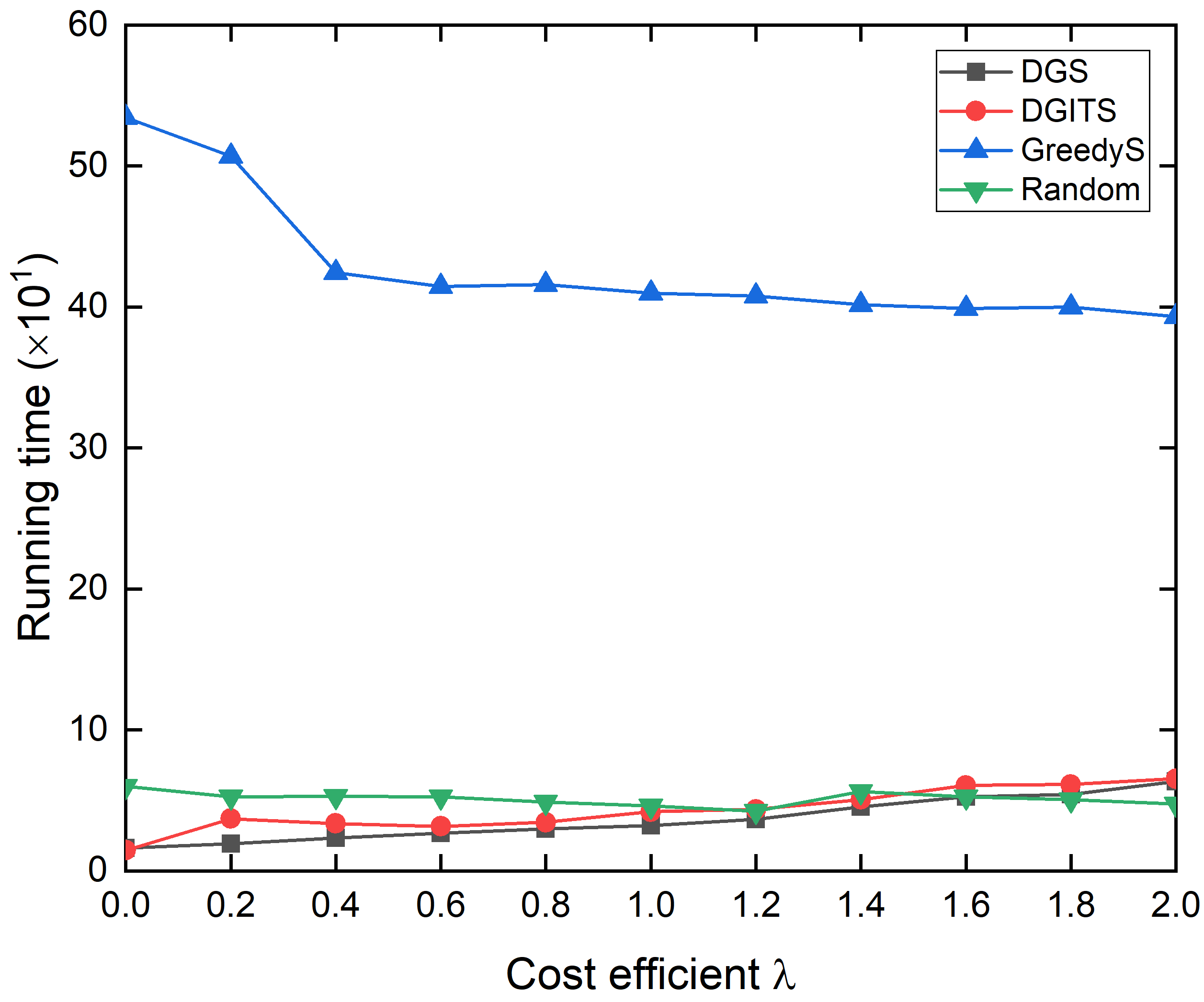}
		%\caption{fig2}
	}%
	
	\subfigure[Epinions, Performance]{
		\centering
		\includegraphics[width=0.48\columnwidth]{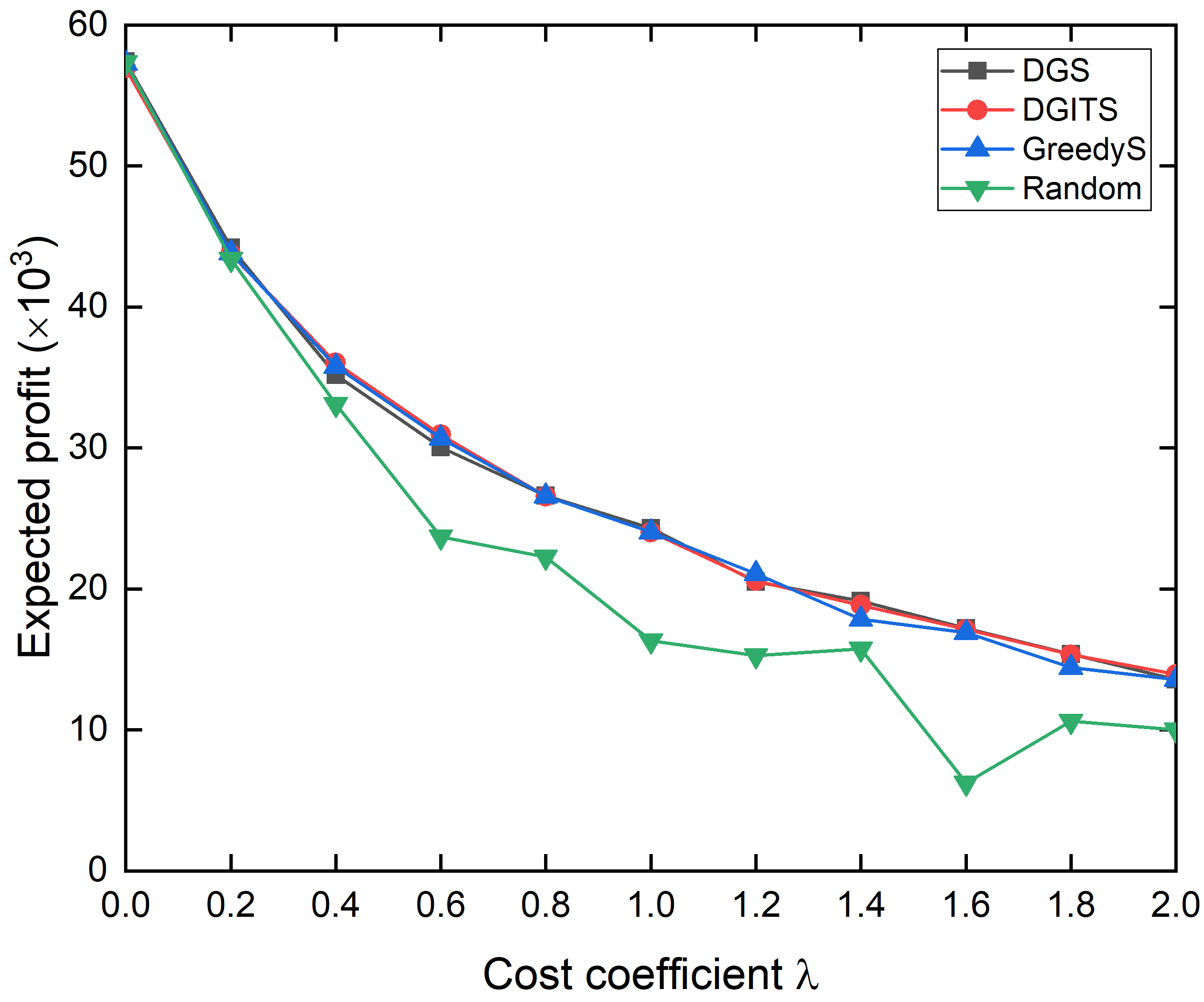}
		%\caption{fig2}
	}%
	\subfigure[Epinions, Time (s)]{
		\centering
		\includegraphics[width=0.48\columnwidth]{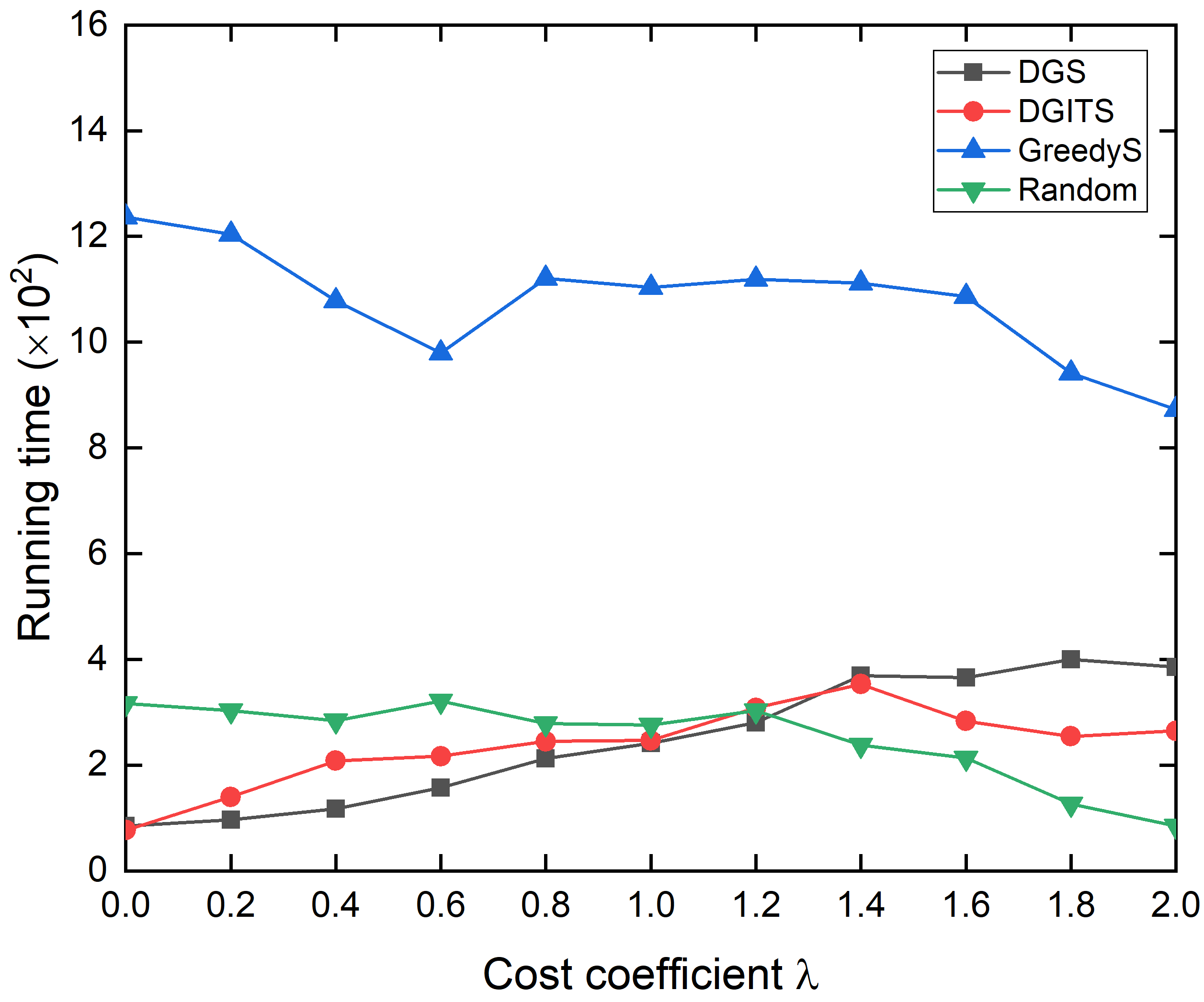}
		%\caption{fig2}
	}%
	\centering
	\caption{The performance and running time comparisons among different algorithms under the LT-model.}
	\label{fig2}
\end{figure}

\subsection{Experimental Settings}
We test different algorithms based on IC/LT-model. For the IC-model, the diffusion probability $p_{uv}$ for each $(u,v)\in E$ is set to the inverse of $v$'s in-degree, i.e., $p_{uv}=1/|N^-(v)|$, and for the LT-model, the weight $b_{uv}=1/|N^-(v)|$ for each $(u,v)\in E$ is set as well, which are adopted by previous studies of IM widely \cite{kempe2003maximizing} \cite{borgs2014maximizing} \cite{tang2014influence} \cite{tang2015influence} \cite{nguyen2016stop}. Then, we need to consider our strategy function, that is
\begin{equation}
	h_u(\vec{x})=1-\prod_{i\in[d]}\prod_{j=1}^{\vec{x}(i)}\left(1-\eta^{j-1}\cdot r_{ui}\right)
\end{equation}
where $\eta\in(0,1)$ is an attenuation coefficient, and $r_{ui}\in[0,1]$ for $u\in V$ and $i\in[d]$, where a unit of investment to marketing action $M_i$ activates user $u$ to be a seed with the probability $r_{ui}$, and each activation is independent. Here, we define vector $\vec{r}_u=(r_{u1},r_{u2},\cdots,r_{ud})$. 

We assume there are five marketing action totally, namely $\vec{x}=(x_1,x_2,\cdots,x_5)$ and $d=5$, and $\vec{b}=\{5\}^d$. Thus, $\vec{x}(i)\leq5$ for each $i\in[5]$. Besides, we set $\eta=0.8$, $\{r_{u1},r_{u3}\}$ is sampled from $[0,0.1]$ and $\{r_{u2},r_{u4},r_{u5}\}$ is sampled from $[0,0.05]$ uniformly. Apparently, $h_u(\vec{x})$ is monotone and dr-submodular with respect to $\vec{x}$. For example, consider a marketing vector $\vec{x}=(1,3,0,0,2)$ and a node $u$ with $\vec{r}_u=(0.1,0.04,0.08,0,0.05)$, we have $h_u(\vec{x})=1-[(1-0.1)][(1-0.04)(1-0.8\times0.04)(1-0.8^2\times0.04)][(1-0.05)(1-0.8\times0.05)]=0.257$ definitely. For the cost function $c$, we adopt a uniform cost distribution. The cost $c_i$ for a unit of marketing action $M_i$, $i\in[d]$, is set as $c_i=\lambda\cdot n/\left\|\vec{b}\right\|_1$, where $\lambda\geq0$ is a cost coefficient. The cost coefficient $\lambda$ defined above is used to regulate the effect of cost on objective function. For example, $f(\cdot)$ is monotone dr-submodular if $\lambda=0$; When we set $\lambda=1$, it implies $f(\vec{b})=0$ if all users in a given social network can be influenced by full marketing vector $\vec{b}$, or else this profit is negative; If $\lambda>1$, we have $f(\vec{b})<0$ definitely.

In addition, the number of Monte-Carlo simulations for each estimation to profit function is $2000$. For those algorithms that adopt speedup by sampling techniques, the parameters setting of four datasets is shown in Table \ref{table2}. Next, we denote ``XXX'' is achieved by Monte-Carlo simulations, but ``XXXS'' is achieved with speedup by sampling techniques. The algorithms we compare in this experiment are shown as follows: (1) DG(S): lattice-based double greedy feed with $[\vec{0},\vec{b}]$; DGIT(S): lattice-based double greedy feed with the collection returned by lattice-based iterative pruning; (3) Greedy(S): select the component with maximum marginal gain until no one has postive gain; (4) Random: select the component randomly until reaching negative marginal gain.
\begin{table}[h]
	\renewcommand{\arraystretch}{1.3}
	\caption{The parameters setting for algorithms that adopt speedup by sampling techniques}
	\label{table2}
	\centering
	\begin{tabular}{|c|c|c|c|c|}
		\hline
		\bfseries Dataset & \bfseries $\varepsilon_1$ & \bfseries $\varepsilon_2$ & \bfseries $\varepsilon_3$ & \bfseries $\delta$\\
		\hline
		NetScience & 0.10 & 0.10 & 0.10 & 10.00\\
		\hline
		Wiki & 0.10 & 0.10 & 0.10 & 10.00\\
		\hline
		HetHEPT & 0.15 & 0.10 & 0.10 & 10.00\\
		\hline
		Epinions & 1.00 & 0.10 & 0.10 & 10.00\\
		\hline
	\end{tabular}
\end{table}

\subsection{Experimental Results}
Fig. \ref{fig1} and Fig. \ref{fig2} draws the expected profit and running time produced by different algorithms under the IC-model and LT-model. From the left columns of Fig. \ref{fig1} and Fig. \ref{fig2}, we can see that the expected profits decrease with the increase of cost coefficient, which is obvious because a larger cost coefficient implies larger cost for a unit of investment. Its trend is close to the inverse proportional relationship, namely $f\propto(1/\lambda)$. Then, the expected profits achieved by DG, DGIT, DGS, and DGITS(DG-IP-RIS) only have very slight even negligible gaps. By comparing the performance between DG and DGS (between DGIT and DGITS), it can show that speedup by sampling techinques is completely effective, which can estimate the objective function accurately. By comparing the performance between DG and DGIT (between DGS and DGITS), it can prove that the optimal solution lies in the shrinked collection returned by itertive pruning, because DGIT does not make the performance of original DG worse. It means that the expected profit will not be reduced at least if we initial double greedy with the shrinked collection returned by iterative pruning. However, do such a thing, it can provide a theoretical bound, so as to avoid some extreme situations. In addition, even if Greedy(S) gives a satisfactory solution in our experiment, there are still some exceptions, for example, (c) in Fig. \ref{fig1} and (g) in Fig. \ref{fig2}. It often happens in some positions with larger cost coefficients.

From the right columns of Fig. \ref{fig1} and Fig. \ref{fig2}, the trend of running time with cost coefficient is a little complex, but there are two apparent characteristics. First, by comparing between DG and DGS (between DGIT and DGITS, or between Greedy and GreedyS), we can see that their running times are reduced significantly by our sampling techniques. Here, in order to test the running time of different algorithms, we do not use parallel acceleration in our implementations. Generally speaking, the running times of algorithms implemented by sampling do not exceed 10\% of the corrsponding algorithms implemented by Monte-Carlo simulations in average. Second, look at DGS and DGITS, their running times increase with the increase of cost coeffient. This is because the lower bound of optimal solution returned by Algorithm \ref{a4} will be smaller and smaller as cost efficient grows, resulting in a larger $\theta_2$ and $\theta_3$. Hence, the number of random RR-sets needed to be generated and searched will increases certainly. Third, by comparing between DGS and DGITS, their running times are roughly equal, shown as (f) (h) in Fig. \ref{fig1} and Fig. \ref{fig2}. It infers that initializing by iterative prunning will not increase the time complexity actually, which is very meaningful.

Table \ref{table3} and Table \ref{table4} shows the effect of lattice-based iterative prunning on the sum of initialized objective values under the IC-model and LT-model, where we denote $A=f(\vec{0})+f(\vec{b})$ and $B=f(\vec{g}^\circ)+f(\vec{h}^\circ)$ for convenience. When cost coefficient $\lambda\geq1$, $A<0$ in all cases, thus there is no approximation guarantee if we run double greedy algorithm feed with $[\vec{0},\vec{b}]$ directly. However, with the help of iterative prunning, $B\geq 0$ holds for most of cases. Like this, our DGIT(S) algorithm is able to offer a $(1/2-\varepsilon)$-approximate solution according to Theorem 5 and Theorem 6.

\begin{table}[h]
	\renewcommand{\arraystretch}{1.3}
	\caption{Sum of initialized objective value under the IC-model}
	\label{table3}
	\centering
	\begin{tabular}{|c|c|c|c|c|c|c|}
		\hline
		\bfseries & \multicolumn{2}{|c|}{NetScience} & \multicolumn{2}{|c|}{Wiki} & \multicolumn{2}{|c|}{HetHEPT} \\
		\hline
		\bfseries $\lambda$ & \bfseries $A$ & \bfseries $B$ & \bfseries $A$ & \bfseries $B$ & \bfseries $A$ & \bfseries $B$ \\
		\hline
		0.8 & -22 & \textbf{219} & -127 & \textbf{379} & -586 & \textbf{7050} \\
		\hline
		1.0 & -97 & \textbf{178} & -303 & \textbf{256} & -2860 & \textbf{848} \\
		\hline
		1.2 & -174 & \textbf{101} & -481 & \textbf{213} & -5065 & \textbf{751} \\
		\hline
		1.4 & -250 & -2 & -658 & \textbf{147} & -7329 & -317 \\
		\hline
		1.6 & -325 & \textbf{82} & -836 & -11 & -9523 & -463 \\
		\hline
		1.8 & -401 & \textbf{55} & -1015 & -269 & -11795 & \textbf{500} \\
		\hline
		2.0 & -477 & -17 & -1192 & -792 & -14031 & -137 \\
		\hline
	\end{tabular}
\end{table}

\begin{table}[h]
	\renewcommand{\arraystretch}{1.3}
	\caption{Sum of initialized objective value under the LT-model}
	\label{table4}
	\centering
	\begin{tabular}{|c|c|c|c|c|c|c|}
		\hline
		\bfseries & \multicolumn{2}{|c|}{NetScience} & \multicolumn{2}{|c|}{Wiki} & \multicolumn{2}{|c|}{HetHEPT} \\
		\hline
		\bfseries $\lambda$ & \bfseries $A$ & \bfseries $B$ & \bfseries $A$ & \bfseries $B$ & \bfseries $A$ & \bfseries $B$ \\
		\hline
		0.8 & 24 & 294 & -64 & \textbf{482} & 1406 & 1410 \\
		\hline
		1.0 & -51 & \textbf{229} & -242 & \textbf{390} & -869 & \textbf{1186} \\
		\hline
		1.2 & -126 & \textbf{33} & -420 & \textbf{177} & -3080 & \textbf{887} \\
		\hline
		1.4 & -202 & \textbf{159} & -599 & \textbf{228} & -5332 & \textbf{392} \\
		\hline
		1.6 & -278 & -6 & -776 & \textbf{166} & -7603 & -948 \\
		\hline
		1.8 & -353 & \textbf{11} & -954 & \textbf{66} & -9807 & -106 \\
		\hline
		2.0 & -429 & \textbf{93} & -1132 & -778 & -12063 & \textbf{229} \\
		\hline
	\end{tabular}
\end{table}

\section{Conclusion}
In this paper, we propose the continuous profit maximization problem first, and based on it, we study unconstrained dr-submodular problem further. For UDSM problem, lattice-based double greedy is an effective algorithm, but there is not approximation guarantee unless all objective values are non-negative. To solve it, we propose lattice-based iterative pruning, and derive it step by step. With the help of this technique, the possibility of satisfying non-negative is enhanced greatly. Our approach can be used as a flexible framework to address UDSM problem. Then, back to CPM-MS problem, we design a speedup strategy by using sampling techniques, which reduce its running time significantly without losing approximation guarantee. Eventually, we evaluate our proposed algorithms on four real networks, and the results validate their effectiveness and time efficiency thoroughly.

% if have a single appendix:
%\appendix[Proof of the Zonklar Equations]
% or
%\appendix  % for no appendix heading
% do not use \section anymore after \appendix, only \section*
% is possibly needed

% use appendices with more than one appendix
% then use \section to start each appendix
% you must declare a \section before using any
% \subsection or using \label (\appendices by itself
% starts a section numbered zero.)
%

% use section* for acknowledgment
\section*{Acknowledgment}

This work is partly supported by National Science Foundation under grant 1747818.

% Can use something like this to put references on a page
% by themselves when using endfloat and the captionsoff option.
\ifCLASSOPTIONcaptionsoff
  \newpage
\fi

% trigger a \newpage just before the given reference
% number - used to balance the columns on the last page
% adjust value as needed - may need to be readjusted if
% the document is modified later
%\IEEEtriggeratref{8}
% The "triggered" command can be changed if desired:
%\IEEEtriggercmd{\enlargethispage{-5in}}

% references section

% can use a bibliography generated by BibTeX as a .bbl file
% BibTeX documentation can be easily obtained at:
% http://mirror.ctan.org/biblio/bibtex/contrib/doc/
% The IEEEtran BibTeX style support page is at:
% http://www.michaelshell.org/tex/ieeetran/bibtex/
%\bibliographystyle{IEEEtran}
% argument is your BibTeX string definitions and bibliography database(s)
%\bibliography{IEEEabrv,../bib/paper}
%
% <OR> manually copy in the resultant .bbl file
% set second argument of \begin to the number of references
% (used to reserve space for the reference number labels box)

\bibliographystyle{IEEEtran}
\bibliography{references}

% biography section
% 
% If you have an EPS/PDF photo (graphicx package needed) extra braces are
% needed around the contents of the optional argument to biography to prevent
% the LaTeX parser from getting confused when it sees the complicated
% \includegraphics command within an optional argument. (You could create
% your own custom macro containing the \includegraphics command to make things
% simpler here.)
%\begin{IEEEbiography}[{\includegraphics[width=1in,height=1.25in,clip,keepaspectratio]{mshell}}]{Michael Shell}
% or if you just want to reserve a space for a photo:

% You can push biographies down or up by placing
% a \vfill before or after them. The appropriate
% use of \vfill depends on what kind of text is
% on the last page and whether or not the columns
% are being equalized.

%\vfill

% Can be used to pull up biographies so that the bottom of the last one
% is flush with the other column.
%\enlargethispage{-5in}

% that's all folks
\end{document}